\documentclass[conference,10pt,letterpaper]{IEEEtran}
\IEEEoverridecommandlockouts  
\usepackage{graphicx}
\usepackage{amssymb}
\usepackage{cite}
\usepackage{algorithm}
\usepackage{algpseudocode}

\newcommand{\remove}[1]{}

\newtheorem{property}{Property}[section]
\newtheorem{lemma}{Lemma}[section]

\renewcommand{\algorithmicforall}{\textbf{for each}}

\newcommand{\stackold}{S_\mathit{old}}
\newcommand{\stacknew}{S_\mathit{new}}

\newcommand{\fid}{\mathit{FID}}
\newcommand{\chains}{\mathit{chains}}
\newcommand{\crossing}{\mathit{crossing}}
\newcommand{\final}{\mathit{final}}
\newcommand{\nexthop}{\mathit{nh}}
\newcommand{\fids}{\mathit{fids}}

\newenvironment{proof}{\textbf{Proof}}{\hfill$\square$}

\begin{document}

\title{Intra-Domain Pathlet Routing}

\author{\IEEEauthorblockN{Marco Chiesa\IEEEauthorrefmark{1},
Gabriele Lospoto\IEEEauthorrefmark{1},
Massimo Rimondini\IEEEauthorrefmark{1}, and
Giuseppe Di Battista\IEEEauthorrefmark{1}}
\IEEEauthorblockA{Roma Tre University\\
Department of Computer Science and Automation\\
\IEEEauthorrefmark{1}\texttt{\{chiesa,lospoto,rimondin,gdb\}@dia.uniroma3.it}}
\thanks{Work partially supported by
ESF project
10-EuroGIGA-OP-003 ``GraDR.''}}

\maketitle

\begin{abstract}
Internal routing inside an ISP network is the foundation for lots of
services that generate revenue from the ISP's customers. A fine-grained
control of paths taken by network traffic once it enters the ISP's network is
therefore a crucial means to achieve a top-quality offer and, equally important,
to enforce SLAs. Many widespread network technologies and approaches (most
notably, MPLS) offer limited (e.g., with RSVP-TE), tricky (e.g., with OSPF
metrics), or no control on internal routing paths. On the other hand, recent
advances in the research community~\cite{bib:pathlet} are a good starting point
to address this
shortcoming, but miss elements that would enable their applicability in an ISP's
network.

We extend pathlet routing~\cite{bib:pathlet} by introducing a new control plane
for internal
routing that has the following qualities: it is designed to operate in the
internal network of an ISP; it enables fine-grained management of network paths
with suitable configuration primitives; it is scalable because routing changes
are only propagated to the network portion that is affected by the
changes; it supports independent configuration of specific network portions 
without the need to know the configuration of the whole network; it is robust 
thanks to the adoption of multipath routing; it supports
the enforcement of QoS levels; it is independent of the specific data plane used
in the ISP's network; it can be incrementally deployed and it can nicely coexist
with
other control planes.
Besides formally introducing the algorithms and messages of our control plane,
we propose an experimental validation in the simulation framework OMNeT++ that
we use to assess the effectiveness and scalability of our approach.
\end{abstract}

\section{Introduction}

It is unquestionable that routing choices inside the network of an Internet
Service Provider (ISP) are critical for the quality of its service offer and, in
turn, for its revenue, and several technologies have been introduced over time to
provide ISPs with different levels of control on their internal routing paths.
These technologies, ranging from approaches as simple as assigning costs to
network links (like, e.g., in OSPF) to real traffic engineering solutions
(like, e.g., RSVP), usually fall short in at least one among: complexity of
setup, predictability of the effects, and degree of control on the routing
paths.
The research community has worked and still contributes to this hot matter 
from different points of view: control over paths is attained by means of source
routing techniques; besides this, many papers advocate the use of
multipath routing as a means to ensure resiliency and quick recovery from
failures; moreover, the granularity of the routing information to be 
disseminated to
support multipath and source routing is sometimes controlled by using 
hierarchical
routing mechanisms.
However, to the extent of our knowledge, existing technological and research
solutions still fail in conjugating a fine-grained control of network paths, 
support for multipath, differentiation of Quality of Service levels, and the 
possibility to independently configure different network portions, a few goals 
that an ISP 
is much interested in achieving without impacting the simplicity of 
configuration primitives, the scalability of the control plane (in terms of 
consumed device memory and of exchanged messages, especially in the presence of 
topological changes), the robustness to faults, and the compatibility with 
existing deployed routing mechanisms.

In this paper we propose the design of a new control plane for internal
routing in an ISP's network which combines all these advantages. Our control 
plane is built on top of pathlet routing~\cite{bib:pathlet}, which we believe to 
be one of the most convenient approaches introduced so far to tackle the ISP's
requirements
described above.

The foundational principles of our control plane are as follows.
A \emph{pathlet} is a path fragment described by a t-uple $\left<\fid, v_1,
v_2, \sigma,\delta\right>$, the semantic being the following: a pathlet,
identified by a value $\fid$, describes the possibility to reach a network node
$v_2$ starting from another network node $v_1$, without specifying any of the
intermediate devices that are traversed for this purpose. A pathlet need not be
an end-to-end path, but can represent the availability of a route from an
intermediate system $v_1$ to an intermediate system $v_2$ in the ISP's network.
An end-to-end path can then be constructed by concatenating several pathlets.
The $\delta$ attribute carries information about the network destinations
(e.g., IP prefixes) that can be reached by using that pathlet (given that
pathlets are not necessarily end-to-end, this attribute can be empty).
In the control
plane we propose, routers are grouped into \emph{areas}: an area is a portion of
the ISP's network wherein routers exchange all information about the available
links, in a much similar way to what a link-state routing protocol does;
however, when announced outside the area, such information is summarized in a
single pathlet that goes from an entry router for the area directly to an exit
router, without revealing routing choices performed by routers that are internal
to the area. This special pathlet, which we call \emph{crossing
pathlet}, is considered outside the area as if it were a single link. An area 
can enclose other areas, thus forming a hierarchical structure with an arbitrary 
number of levels:
the $\sigma$ attribute in a pathlet encodes a restriction about the
areas where that pathlet is supposed to be visible.

In designing our control plane we took into account several aspects, among
which: efficient reaction to topological changes and administrative
configuration changes, meaning that the effects of such changes are only
propagated to the network portion that is affected by them; support for several
kinds of routing policies; support for multipath and differentiation of QoS 
levels; and compatibility
and integration with other technologies that are already deployed in the ISP's
network, to allow an incremental deployment. By introducing areas we also offer 
the 
possibility for different network administrators to independently configure 
different portions of an ISP's network without the need to be aware of the 
overwhelmingly complex setup of the whole network.

Our contribution consists of several parts. First of all, we introduce a
model for a network where nodes are grouped in a hierarchy of areas. Based on
this model, we define the basic mechanisms adopted in the creation and
dissemination of pathlets in the network. We then present a detailed description
of how network dynamics are handled, including the specification of the messages
of our control plane and of the algorithms executed by a network node upon
receiving such messages or detecting topological or configuration changes.
Further, we elaborate on the practical applicability of our control plane in an
ISP's network in terms of possible deployment technologies and propose some
possible extensions to accommodate further requirements. Last, we present an
experimental assessment of the scalability of our approach in a simulated
scenario.

The rest of the paper is organized as follows.
In Section~\ref{sec:related-work} we review and classify the state of the art on
routing mechanisms that could match the requirements of ISPs.
In Section~\ref{sec:model} we introduce our formal network model.
In Section~\ref{sec:dissemination-basics} we describe the basic
pathlet creation and dissemination mechanisms.
In Section~\ref{sec:dynamics} we detail the message types of our control
plane and describe the network dynamics.
In Section~\ref{sec:applicability} we present applicability considerations
and possible extensions to accommodate other requirements.
In Section~\ref{sec:experiments} we present the results of our experiments run
in the OMNeT++ simulation framework.
Last, conclusions and plan for future work are presented in
Section~\ref{sec:conclusions}\mbox{.}

\section{Related Work}\label{sec:related-work}
Many of the techniques that we adopt in our control plane have already been
proposed in the literature. Most notably, these techniques include source
routing (intended as the possibility for the sender of a packet to select the
nodes that the packet should traverse), hierarchical routing (intended as a
method to hide the details of routing paths within certain portions of the
network by defining areas), and multipath routing (intended as the possibility
to compute and keep multiple paths between each source-destination pair).
However, none of the contributions we are aware of combines them in a way that
provides all the benefits offered by our approach. We provide
Table~\ref{tab:state-of-the-art} as a reading key to compare the state of the
art on relevant control plane mechanisms, discussed in the following.

\begin{table}[t]
	\centering
   \begin{tabular}{|l|c|c|c|}
		\cline{2-4}
      \multicolumn{1}{l|}{} & \textbf{Source routing} & \textbf{Hierarchical
routing} & \textbf{Multipath routing}\\
		\hline
		\cite{bib:path-splicing} & Limited & No & Yes\\
		\hline
		\cite{bib:nira} & Yes & No & No\\
		\hline
		\cite{bib:hlp} & No & Limited & No\\
		\hline
		\cite{bib:macro-routing} & No & Yes & No\\
		\hline
		\cite{bib:hdp} & Limited & Yes & No\\
		\hline
		\cite{bib:slick-packet} & Yes & Limited & Yes\\
		\hline
		\cite{bib:bgp-add-paths} & No & No & Yes\\
		\hline
		\cite{bib:hierarchical-routing-lv} & No & Yes & Limited\\
		\hline
		\cite{bib:landmark} & No & Yes & Yes\\
		\hline
		\cite{bib:miro} & Limited & No & Yes\\
		\hline
		\cite{bib:yamr} & No & Limited & Yes\\
		\hline
   \end{tabular}
   \caption{A classification of the state of the art according to the adoption
of some relevant routing techniques.}
   \label{tab:state-of-the-art}
\end{table}

Pathlet routing~\cite{bib:pathlet} is probably the contribution that is
closest to our control plane approach: its most evident drawback is the lack of
a clearly defined mechanism for the dissemination of pathlets, which the
authors only hint at.
Path splicing~\cite{bib:path-splicing} is a mechanism designed with
fault tolerance in mind (see also~\cite{bib:pathsplicing-fault-tolerance}): it
exploits multipath to ensure connectivity between network nodes as long as the
network is not partitioned. However, actual routing paths are not exposed, and
this limits the control that the ISP could enforce on internal routing. Even in
MIRO~\cite{bib:miro}, where multiple paths can be negotiated to satisfy the
diverse requirements of end users, there can be no full control of a whole
routing path.
NIRA~\cite{bib:nira} compensates this shortcoming, but it is designed only for
an interdomain routing architecture, like MIRO, and it relies on a
constrained address space allocation, a hardly feasible choice for an ISP that
is taken also by Landmark~\cite{bib:landmark}.
Slick packets~\cite{bib:slick-packet} is also designed for fault tolerant source
routing, achieved by encoding in the forwarded packets a directed acyclic graph
of different alternative paths to reach the destination. Besides the intrinsic
difficulty of this encoding, it inherits the limits of the dissemination
mechanisms it relies on: NIRA or pathlet routing. BGP
Add-Paths~\cite{bib:bgp-add-paths} and YAMR~\cite{bib:yamr} also address
resiliency by announcing multiple paths selected according to different
criteria, but they only adopt multipath routing, provide very limited or no
support for hierarchical routing, and have some dependencies on the BGP
technology.
A completely different approach is taken by HLP~\cite{bib:hlp}, which proposes a
hybrid routing mechanism based on a combination of link-state and path-vector
protocols. This paper also presents an in-depth discussion of routing
policies that can be implemented in such a scenario. Although this
contribution matches more closely our approach, it is not conceived for
internal routing in an ISP's network, it constrains the way in which areas are
defined on the network, and it has limits on the configurable routing
policies. A similar hybrid routing mechanism called
ALVA~\cite{bib:hierarchical-routing-lv} offers more flexibility in the
configuration of areas but, like Macro-routing~\cite{bib:macro-routing},
it does not explicitly envision source routing and multipath routing.
HDP~\cite{bib:hdp} is a variant of this approach that, although natively
supporting Quality of Service and traffic engineering objectives, is closely
bound to MPLS and accommodates source routing and multipath routing only in the
limited extent allowed by this technology.

Some of the papers we mention here also point out an aspect that is key to
attain the nice control plane features we are looking for: path-vector protocols
allow the setup of complex information hiding and manipulation policies, whereas
link-state protocols offer fast convergence with a low overhead. Therefore,
a suitable combination of the two mechanisms, which is considered in our
approach, should be pursued to inherit the advantages of both.

\section{A Hierarchical Network Model}\label{sec:model}
We now describe the hiearchical model we use to represent the network.
\begin{figure}
   \centering
   \includegraphics[width=4.5cm]{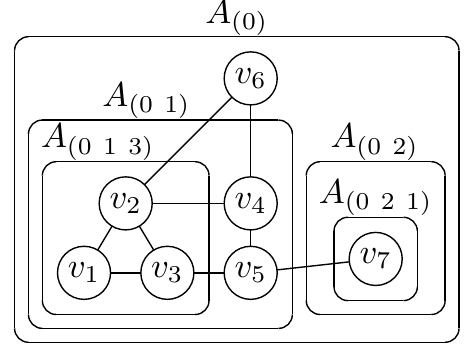}
   \caption{A sample network. Stack labels are integer numbers. Rounded boxes
represent areas $A_\sigma$, with the associated stacks specified as subscript
$\sigma$.}
   \label{fig:example}
\end{figure}
We model the physical network topology as a graph, with vertices being
routers and edges representing links between routers: let
$G=(V,E)$ be an undirected graph, where $V$ is a set of vertices and
$E=\{(u,v)|u,v\in V\}$ is a set of edges that connect vertices.
Fig.~\ref{fig:example} shows an example of such graph.
We assume that any vertex in the graph is interested in establishing a
path to special vertices that represent routers that announce network
destinations. Therefore, we introduce a set of \emph{destination
vertices} $D\subseteq V$.
We highlight that the same representation can be adopted to capture the topology
of overlay networks, while keeping the model unchanged.

In order to improve scalability and limit the propagation of routing
information that is only relevant in certain portions of the network, we group
vertices into structures called areas. To describe the assignment of a vertex
$v\in V$ to an area we associate to $v$ a stack of labels $S(v)=(l_0\ l_1\
\dots\ l_n)$, where each label is taken from a set $L$. To simplify notation
and further reasoning, we
assume that $l_0$ is the same for every
$S(v)$ and that $\bot\in L$. $\bot$ is a special label that we will use to
identify routing information (actually, pathlets) that represents network links.
Referring to the example in Fig.~\ref{fig:example}, we have
$L=\{0, 1, 2, 3, \bot\}$,
$S(v_1)=S(v_2)=S(v_3)=(0\ 1\ 3)$,
$S(v_4)=S(v_5)=(0\ 1)$, $S(v_6)=(0)$, and $S(v_7)=(0\ 2\ 1)$.

We now define some operations on label stacks that allow us to introduce the
notion of area and will be useful in the rest of the paper.
Given two stacks
$\sigma_1=(l_1\ l_2\ \dots\ l_i)$ and
$\sigma_2=(l_{i+1}\ l_{i+2}\ \dots\ l_n)$, we define their concatenation as
$\sigma_1 \circ \sigma_2=(l_1\ l_2\ \dots\ l_i\ l_{i+1}\ l_{i+2}\ \dots\ l_n)$.
Assuming that $()$ indicates the empty stack, we have that $\sigma \circ ()=()
\circ \sigma = \sigma$. Given two stacks $\sigma_1$ and $\sigma_2$, we say that
$\sigma_2$ \emph{strictly extends} $\sigma_1$, denoted by $\sigma_1\sqsubset
\sigma_2$, if $\sigma_2$ is longer than $\sigma_1$ and $\sigma_2$ starts
with the same sequence of labels as in $\sigma_1$, namely there exists a
nonempty stack $\bar\sigma$ such that
$\sigma_2=\sigma_1 \circ \bar\sigma$. We say that $\sigma_2$ \emph{extends}
$\sigma_1$, indicated by $\sigma_1\sqsubseteq\sigma_2$, if $\bar\sigma$ can be
empty.


We call \emph{area} $A_\sigma$ a non-empty set of
vertices whose stack extends $\sigma$,
namely a set $A_\sigma\subseteq V$ such that $\forall v\in A_\sigma:
\sigma\sqsubseteq S(v)$. The following property is a consequence of this
definition:
\begin{property}\label{pro:belonging-areas}
Given a vertex $v\in V$ with stack $S(v)$, $v$ belongs to the
following set of areas: $\{A_\sigma | \sigma\sqsubseteq S(v)\}$.
\end{property}
Our definition of area has a few interesting consequences.
First, by Property~\ref{pro:belonging-areas}, specifying the stack $S(v)$ for a
vertex $v$ defines all areas $A_\sigma$ such that $\sigma\sqsubseteq S(v)$.
Thus, areas can be conveniently defined by simply specifying the label stacks
for all vertices.
Considering again the example in Fig.~\ref{fig:example}, the
assignment of label stacks to vertices implicitly
defines areas $A_{(0)}$, $A_{(0\ 1)}$, $A_{(0\ 1\ 3)}$, $A_{(0\ 2)}$, and
$A_{(0\ 2\ 1)}$ (note that $A_{(0\ 2)}=A_{(0\ 2\ 1)}$).
Moreover, areas can contain other areas, thus forming a hierarchical
structure. However, areas can never overlap partially, that is, given any two
areas $A_1$ and $A_2$, it is always $A_1\subseteq A_2$ or $A_2\subseteq A_1$.
Also, the first label $l_0$ in any stack plays a special role, because it is:
$A_{(l_0)}=V$.


Areas are introduced to hide the detailed internal topology of portions of
the network and, therefore, to limit the scope of propagation of routing
information.  As a general rule, assuming that the internal topology of an area
$A_\sigma$ consists of all the vertices in $A_\sigma$ and the edges of $G$
connecting those vertices, our control plane propagates only a
summary of this information to vertices outside $A_\sigma$.
With this approach in mind, we introduce two additional operators on
label stacks, that are used to determine the correct level of granularity
to be used in propagating routing information.
Given two areas $A_{\sigma_a}$ and $A_{\sigma_b}$, the first operator, indicated
by $\Join$, is used to determine the most nested area that contains both
$A_{\sigma_a}$ and $A_{\sigma_b}$, namely the area within which routing
information that is relevant only for vertices in $A_{\sigma_a}$ and
$A_{\sigma_b}$ is supposed to be confined: this area is defined by
$A_{\sigma_a\Join\sigma_b}$. Referring to the example
in Fig.~\ref{fig:example}, the most nested area containing both $v_5$ and $v_7$
is $A_{S(v_5)\Join S(v_7)}=A_{(0)}$.
The second operator, indicated by $\rightarrowtail$, is used to determine the
least nested area that includes all vertices in $A_{\sigma_a}$ but not those in
$A_{\sigma_b}$, namely the area that vertices in $A_{\sigma_a}$ declare to be
member of when sending routing information to neighboring vertices in
$A_{\sigma_b}$: this area is defined by $A_{\sigma_a\rightarrowtail\sigma_b}$
(in case $\sigma_a\sqsubseteq\sigma_b$, such an area does not exist and
$A_{\sigma_a\rightarrowtail\sigma_b}=A_{\sigma_a}$).
Considering again Fig.~\ref{fig:example}, $v_7$ communicates with $v_5$ as a
member of area $A_{S(v_7)\rightarrowtail S(v_5)}=A_{(0\ 2)}$.
We now define the two operators formally. Given two arbitrary stacks
$\sigma_a=(a_0\ \dots\ a_i\ \dots\ a_n)$ and
$\sigma_b=(b_0\ \dots\ b_i\ \dots\ b_m)$ such that $a_0=b_0$, $a_1=b_1$,
$\dots$, $a_i=b_i$ for some $i\leq \min(m,n)$ and $a_{i+1}\neq b_{i+1}$ if
$i<\min(m,n)$,
we define $\sigma_a \Join\sigma_b=(a_0\ \dots\ a_i)$ and
$\sigma_a \rightarrowtail \sigma_b=(a_0\ \dots\ a_k)$ where
$k=\min(i+1,n)$.
We extend these definitions in a natural way by assuming that
$()\Join\sigma_b=\sigma_a\Join()=()$ and
$\sigma_a\rightarrowtail()=()\rightarrowtail\sigma_b=()$.
Be aware that $\Join$ is commutative, whereas
$\rightarrowtail$ is not.
For each area, a subset of the vertices belonging to the area
are in charge of summarizing internal routing information and propagating it
outside the area: these vertices are called border vertices.
In particular, a vertex $u\in A_\sigma$ incident on an edge $(u,v)$ such that
$v\notin A_\sigma$ is called a \emph{border vertex for area
$A_\sigma$}. In the example in Fig.\ref{fig:example}, $v_2$ is a border vertex
for area $A_{(0\ 1\ 3)}$ because $v_2\in A_{(0\ 1\ 3)}$, $(v_2,v_6)\in E$, and
$v_6\notin A_{(0\ 1\ 3)}$.
Because of Property~\ref{pro:belonging-areas}, a single vertex can be a border
vertex for more than one area: in Fig.~\ref{fig:example}, $v_2$ is also a border
vertex for area $A_{(0\ 1)}$ because $v_2\in A_{(0\ 1)}$ and $v_6\notin A_{(0\
1)}$. Also, by definition it may be the case that a
neighbor of a border vertex is not a border vertex for any areas: looking again
at Fig.~\ref{fig:example}, $v_6$ is not a border vertex.
Derived from the definition of border vertex, we can state the following
property:

\begin{property}\label{pro:no-border}
   There can be no border vertex for area $A_{(l_0)}$.
\end{property}

\section{Basic Mechanisms for the Dissemination
of Routing Information}\label{sec:dissemination-basics}

After introducing our network model, we can now illustrate how routing
information is disseminated over the network. In order to do so, we first
define the concept of pathlet and describe how pathlets are created and
propagated. We then introduce conditions on label stacks and routing policies 
that regulate the propagation of pathlets.

%
%
%


\vspace{1mm}
{\bf Pathlets} -- 
In order to learn about paths to the various destinations, vertices in graph $G$
exchange path fragments called pathlets~\cite{bib:pathlet}. In order to support
the definition of areas and the consequent information hiding mechanisms, we
present an enhanced definition of a pathlet that is
slightly different from the original one.  A \emph{pathlet} $\pi$ is a t-uple
$\left<\fid, v_1, v_2, \sigma,\delta\right>$ where all fields are assigned
by vertex $v_1$: $\fid$ is an identifier of the pathlet called \emph{forwarding
identifier}, and is unique at $v_1$; $v_1\in V$ is the
\emph{start vertex}; $v_2\in
V|v_2\neq v_1$ is the \emph{end vertex}; $\sigma$ is a stack of labels from
$L$ called \emph{scope stack}, and is a new field introduced
to restrict the areas where pathlet $\pi$ should be propagated; and $\delta$ is 
a (possibly empty)
set of network destinations (e.g., network prefixes) available at $v_2$.
$\fid$s are used to distinguish between different pathlets starting at the
same vertex $v_1$ and are exploited
by the data plane of $v_1$ to determine where traffic is to be
forwarded.
Even pathlets that have the same scope stack and, using different network 
paths, connect the same pair of vertices, can still be distinguished based on 
the $\fid$.
We assume $\fid$s are integer numbers.
%
%

\vspace{1mm}
{\bf Packet forwarding} --
Each vertex has to keep forwarding state information to support the
operation of the data plane. Since our control plane has to update these
information, we now define the forwarding state of a vertex by providing hints
about the packet forwarding mechanism, which is the same
presented in~\cite{bib:pathlet}.
In pathlet routing, each data packet carries in a dedicated header a sequence
of $\fid$s: this sequence indicates the pathlets that the packet should be
routed along to reach the destination. When a vertex $u$ receives a packet, it
considers the first $\fid$ in the sequence contained in its header: this
$\fid$, referenced as $f$ in the following, uniquely identifies a pathlet $\pi$
that is known at $u$ and that has $u$ as start vertex.
Now, in the general case pathlet $\pi$ may lead to an end vertex that is not
adjacent to $u$. Since a pathlet does not contain the detailed specification of
the routing path to be taken to reach the end vertex, before forwarding the 
packet $u$ has
to modify the sequence of $\fid$s contained in the packet header to insert such 
specification: $u$ achieves this by replacing
$f$ with another sequence of $\fid$s that indicates the pathlets to be used to
reach the end vertex of $\pi$. Therefore, the first part of the forwarding
state of $u$ is a correspondence between each value of the $\fid$ and a (possibly
empty) sequence of $\fid$s, which we indicate as $\fids_u(\fid)$.
At this point, $u$ has to pick a neighboring vertex to forward the packet to. 
Since 
also this information is missing in pathlet $\pi$, it must be
kept locally at vertex $u$. The second part of the forwarding
state of $u$ is therefore the specification of the \emph{next-hop} vertex, 
namely of
the vertex that immediately comes after $u$ along $\pi$, which we refer to as
$\nexthop_u(\fid)$.
%
%
%
Both $fids_u$ and $\nexthop_u$ are computed by the control plane, as
explained in the following section.


%

\vspace{1mm}
{\bf Atomic, crossing, and final pathlets} --
We distinguish among three types of pathlets: atomic, crossing, and final.
A pathlet $\pi=\left<\fid, v_1, v_2, \sigma, \delta\right>$ is called
\emph{atomic pathlet} if its start and end vertices are
adjacent on graph $G$. Atomic pathlets carry in the $\delta$ field the network 
destinations possibly available at $v_2$. They are used to propagate information 
about the network topology and are propagated only inside the most nested area 
that contains both $v_1$ and $v_2$. To represent the fact that a network link 
$(v_1,v_2)$ is bidirectional, two atomic pathlets need to be created for that 
link, one from $v_1$ to $v_2$ (created by $v_1$) and another from $v_2$ to $v_1$ 
(created by $v_2$). Atomic pathlets are always marked by putting
the special label $\bot$ at the end of the scope stack. More formally, an atomic
pathlet is such that $(v_1,v_2)\in E$ and
$\exists\bar\sigma\neq()|\sigma = \bar\sigma\circ\bot$.
Besides serving as a distinguishing mark for atomic pathlets, label
$\bot$ has been introduced to simplify the description of pathlet
dissemination mechanisms, because it avoids the need to consider several special 
cases.

Pathlet $\pi$ is a \emph{crossing pathlet} for area $A_\sigma$ if its
start and end vertices are border vertices for area $A_\sigma$. Crossing 
pathlets always have $\delta=\emptyset$ and do not contain label $\bot$ in the 
scope stack. A pathlet of this type offers vertices 
outside
$A_\sigma$ (that is, 
whose label stack is strictly extended by $\sigma$) the possibility to traverse 
$A_\sigma$ without knowing its internal
topology: crossing pathlets are therefore one of the fundamental building blocks 
of our control plane, as they realize the possibility to hide detailed routing 
information about the interior of an area.
Since a vertex can be a border vertex for more than one area, different pathlets
with the same start and end vertices can act as crossing pathlets for different
areas (they would have different scope stacks and $\fid$s).

Last, $\pi$ is a \emph{final pathlet} if it leads to some network destination 
available at $v_2$, 
that is, if $\delta\neq\emptyset$. Like crossing pathlets, final pathlets do not 
contain label $\bot$ in the scope stack. Final pathlets are created by a border 
vertex 
$v_1$ for an area $A_\sigma$ to inform vertices outside $A_\sigma$ about the 
possibility to 
reach a destination vertex $v_2\in A_\sigma\cap D$.

Notice that between two neighboring vertices it possible to create an atomic, a 
crossing, and a final pathlet: these pathlets are disseminated independently 
and have each a different role, as described above. The type (and, therefore, 
the 
role and scope of propagation) of these pathlets can be determined
based on the contents of $\delta$ and on the presence of the special label 
$\bot$ in the scope stack.
Since the creation and dissemination mechanisms are very similar for crossing 
and 
final pathlets, in the following we detail only those applied to crossing 
pathlets, assuming that they are the same for final pathlets unless differently 
stated.


\vspace{1mm}
{\bf Pathlet creation} --
We now describe how atomic and crossing pathlets are created at each vertex 
(similar mechanisms are applied for final pathlets). 
When we say ``create'' we mean that a vertex defines these pathlets, assigns to 
each of them a unique $\fid$, and keeps them in a local data structure, as 
illustrated in Section~\ref{sec:dynamics}. In the following, we also use the 
term ``composition'' to refer to the creation of crossing and final pathlets.

Each vertex $u\in V$ creates atomic pathlets $\left<\fid,u,v,\sigma\circ 
(\bot),\delta\right>$ such that $(u,v)\in E$, $\sigma=S(u)\Join S(v)$, and 
$\delta$ contains the set of network destinations possibly available at $v_2$. 
The 
scope stack $\sigma$ is chosen in such a way to restrict propagation of each 
atomic 
pathlet up to the most nested area that contains both $u$ and $v$.
These pathlets are used to disseminate information about the physical
network topology and act as building blocks for creating crossing and final
pathlets. When creating an atomic pathlet, vertex $u$ also updates its 
forwarding state with $\nexthop_u(\fid)=v$ 
and
$\fids_u(\fid)=()$.
Looking at the example of Fig.~\ref{fig:example}, $v_4$ creates 
pathlets
$\left<1,v_4,v_5,(0\ 1\ \bot),\emptyset\right>$,
$\left<2,v_4,v_2,(0\ 1\ \bot),\emptyset\right>$, and
$\left<3,v_4,v_6,(0\ \bot),\emptyset\right>$ (we assigned $\fid$s randomly).
$v_4$ 
then sets $\nexthop_{v_4}(1)=v_5$, $\nexthop_{v_4}(2)=v_2$, 
$\nexthop_{v_4}(3)=v_6$, and $\fids_{v_4}(1)=\fids_{v_4}(2)=\fids_{v_4}(3)=()$.

Atomic pathlets can be concatenated to create pathlets between
non-neighboring vertices.
To achieve this, we introduce a set $\chains(\Pi,u,v,\sigma)$ that contains all 
the possible concatenations of pathlets taken from a set $\Pi$, that start at 
$u$ and end at $v$, and whose scope stack extends $\sigma$, regardless of 
$\fid$s and network destinations.
$\chains(\Pi,u,v,\sigma)$ is formally defined as the set of all
possible sequences of pathlets in $\Pi$, where each sequence
$(\pi_1\ \pi_2\ \dots\ \pi_n)$ is finite, cycle-free, and such that
$\pi_i=\left<
\fid_i,w_i,w_{i+1},\sigma_i,\delta_i\right>$, $\sigma\sqsubseteq\sigma_i$,
$\pi_{i+1}=\left<\fid_{i+1},w_{i+1},w_{i+2},\sigma_{i+1},\delta_{i+1}\right>$,
and
$\sigma\sqsubseteq\sigma_{i+1}$, with $i\in\{1, \dots, n-1\}$,
$w_1=u$, $w_{n+1}=v$.
A border vertex $u$ exploits these concatenations to create crossing pathlets, 
that can be used to traverse the areas that $u$ belongs to as if they consisted 
of a single link. Although $u$ may be a border vertex for several 
areas, it creates crossing pathlets only for those areas that $u$'s 
neighbors are actually interested in traversing. To find out which are these 
areas, we 
must consider how $u$ appears to its neighbors: we assume that each neighbor 
$n$ of $u$ that is not in $A_{S(u)}$ considers $u$ as a member of the least 
nested area 
that includes $u$ but not $n$, that is, area $A_{S(u)\rightarrowtail S(n)}$. For 
this reason, $u$ creates a set of crossing pathlets for each area $\bar
A = A_{S(u)\rightarrowtail S(n)}$: these pathlets start at $u$ and end at any 
other border vertex $v$ for $\bar A$, $v\neq u$. Similarly, $u$ creates final 
pathlets that start at $u$ and end at any other destination vertex $v\in 
D\cap\bar A$.
In the example in Fig.~\ref{fig:example}, $v_6$ considers $v_2$ as a member of 
area
$A_{(0\ 1\ 3)\rightarrowtail (0)=(0\ 1)}$, whereas $v_4$ considers $v_2$ as a 
member of area $A_{(0\ 1\ 3)\rightarrowtail (0\ 1)=(0\ 1\ 3)}$. For this reason, 
$v_2$ will create crossing and final pathlets for $A_{(0\ 1)}$ to be offered to 
$v_6$ and 
crossing and final pathlets for $A_{(0\ 1\ 3)}$ to be offered to $v_4$.
%
%
%
%
More formally, for each neighbor $n$, a border vertex $u\in A_\sigma$ creates
crossing pathlets by populating a set $\crossing_u(\Pi,\sigma)$, with
$\sigma=S(u)\rightarrowtail S(n)$. Each set $\crossing_u(\Pi,\sigma)$
contains a pathlet 
$\pi=\left<\fid,u,w,\sigma,\delta\right>$ for each border vertex $w\neq u$ for 
$A_\sigma$ and for each sequence 
$(\pi_1\ \pi_2\ \dots\ \pi_n)$ in set $\chains(\Pi,u,w,\sigma)$. $\fid$ is 
chosen in such 
a way to be unique at $u$ and $\delta$ is set to the empty set $\emptyset$.
Assuming that $\pi_i=\left<\fid_i,u_i,v_i,\sigma_i,\delta_i\right>$, the 
forwarding state of 
$u$ is updated by setting $\fids_u(\fid)=(\fid_2\ \fid_3\ \dots \ \fid_n)$ and 
$\nexthop_u(\fid)=\nexthop_u(\fid_1)=v_1$. Note that, in general, pathlet 
$\pi_1$ may not be an atomic pathlet: in this case, $u$ has to 
recursively expand $\pi_1$ into the component atomic pathlets in order to get 
the 
correct 
sequence of $\fid$s to be put in $\fids_u(\fid)$ and the correct next-hop to be 
assigned as $\nexthop_u(\fid)$. However, because of the way in which set 
$\chains(\Pi,u,w,\sigma)$ will be used in the following, and in particular 
because of the composition of set $\Pi$ on which it will be constructed, we 
assume without loss of generality that $\pi_1$ is always an atomic pathlet.
As an example taken from Fig.~\ref{fig:example}, let
$\Pi=\{\left<2,v_2,v_4,(0\ 1\ \bot),\emptyset\right>,$
$\left<3,v_4,v_5,(0\ 1\ \bot),\emptyset\right>,$
$\left<1,v_1,v_3,(0\ 1\ 3\ \bot),\emptyset\right>,$
$\left<2,v_3,v_5,(0\ 1\ \bot),\emptyset\right>\}$. $v_2$ may have in its set 
$\crossing_{v_2}(\Pi, (0\ 1))$ a pathlet
$\left<1,v_2,v_5,(0\ 1),\emptyset\right>$ corresponding to the sequence of 
atomic pathlets
$(\left<2,v_2,v_4,(0\ 1\ \bot),\emptyset\right>\ \left<3,v_4,v_5,(0\ 1\ 
\bot),\emptyset\right>)$ taken from set $\chains(\Pi,v_2,v_5,(0\ 1))$. $v_2$ 
will therefore set $\fids_{v_2}(1)=(3)$ and $\nexthop_{v_2}(1)=v_4$.



Final pathlets are created in a much similar way as crossing pathlets, except 
that 
they are composed towards vertices in $A_\sigma\cap D$ and $\delta$ is set to 
the set $\delta_n$ of network destinations of the last 
component pathlet in the sequence. Final pathlets are put in a set 
$\final_u(\Pi,\sigma)$.


Because of the way in which pathlets are created and of the fact that there are 
no
crossing or final pathlets for area $A_{(l_0)}$ (Property~\ref{pro:no-border}), 
we can easily conclude that there are always at least two labels in the scope 
stack of any pathlet. This is stated by the following property:
\begin{property}\label{pro:nonempty-scope-stack}
   For any pathlet $\left< \fid,u,v,\sigma,\delta\right>$ there exists
$\bar\sigma\neq ()$ such that $\sigma=(l_0)\circ\bar\sigma$.
\end{property}



\vspace{1mm}
{\bf Discovery of border vertices} --
In order to be able to compose crossing pathlets for an area, a border
vertex $u$ must be able to discover which are the other border vertices for the
same area. The only information that $u$ can exploit to this
purpose are the pathlets it has received. Given that a border vertex connects
the inner part of an area with vertices outside that area, a simple
technique to detect whether a vertex $v$ is a border vertex
consists therefore in comparing the scope stacks of suitable pairs of
pathlets that have $v$ as a common vertex.

The technique is based on the following lemma.
\begin{lemma}
If a vertex $u\in A_\sigma$ receives two pathlets
$\pi_1=\left<\fid_1,v_1,w_1,\sigma_1\circ (l),\delta_1\right>$ and
$\pi_2=\left<\fid_2,v_2,w_2,\sigma_2\circ (\bot),\emptyset\right>$, with
$l\in L$, $\sigma_1\neq ()$, $\sigma_2\neq ()$, the start and end vertices of 
$\pi_1$ and $\pi_2$ are such that $v_1\neq v_2$ or $w_1\neq w_2$, the scope 
stacks of $\pi_1$ and $\pi_2$ are
such that $\sigma_2\sqsubset \sigma=\sigma_1$, and there exists a vertex $v$ 
such that both $\pi_1$ and $\pi_2$ start or end at $v$, then $v$ is a border 
vertex 
for
$A_\sigma$.
\end{lemma}
\begin{proof}
The statement follows from the way in which scope stacks are assigned to
pathlets. The fact that $v\in\{v_1,w_1\}$ implies that $v\in A_{\sigma_1}$:
in fact, if $l=\bot$, then $\pi_1$ is an atomic pathlet whose scope stack is
therefore assigned in such a way that $\sigma_1=S(v_1)\Join S(w_1)$; since we
know that $S(v_1)\Join S(w_1)\sqsubseteq S(v)$, by
Property~\ref{pro:belonging-areas} we can conclude that $v\in A_{\sigma_1}$.
Otherwise, if $l\neq\bot$, then $\pi_1$ is either a crossing pathlet for some
area $A_{\sigma_1\circ (l)}$ or a final pathlet; in both cases, being an 
endpoint of pathlet 
$\pi_1$, $v$ must belong to $A_{\sigma_1\circ (l)}$ and, using 
Property~\ref{pro:belonging-areas} again, this also implies that $v\in 
A_{\sigma_1}$.
Since $\sigma_1=\sigma$, we can conclude that $v\in
A_\sigma$. On the other hand, from the scope stack 
$\sigma_2\sqsubset\sigma$ of the atomic pathlet $\pi_2$ we know that $v$
has some neighbor that is not in $A_{\sigma}$: this makes $v$ a border vertex
for $A_\sigma$.
%
\end{proof}
According to this lemma, a vertex $u\in A_\sigma$ can use the following simple
algorithm, formalized as function \textsc{DiscoverBorderVertices}($u$, $\sigma$, 
$\Pi$) in Algorithm~\ref{alg:detect-border-vertices}, to discover other border 
vertices for $A_\sigma$ based on a set of known pathlets $\Pi$:
consider any possible pair $(\pi_1,\pi_2)$ of pathlets in $\Pi$ whose start and 
end vertices have exactly one vertex $v$ in common; if this pair satisfies the 
conditions of the
lemma, $v$ is a border vertex for $A_\sigma$.

\begin{algorithm}[t]
   \caption{Algorithm that a vertex $u\in A_\sigma$ can use to discover remote 
   border vertices for $A_\sigma$ based on the known pathlets in $\Pi$.}
   \label{alg:detect-border-vertices}
   \begin{algorithmic}[1]
      \Function{DiscoverBorderVertices}{$u$, $\sigma$, $\Pi$}
         \State $B\leftarrow\emptyset$
         \ForAll{pair $(\pi_1,\pi_2)$ of pathlets with 
         $\pi_1=\left<\fid_1,v_1,w_1,\sigma_1,\delta_1\right>$ and 
         $\pi_2=\left<\fid_2,v_2,w_2,\sigma_2,\delta_2\right>$, 
         such that $v_1\neq v_2$ or $w_1\neq w_2$, and $\exists v$ 
         such that both $\pi_1$ and $\pi_2$ start or end at $v$}
            \If{$\exists\bar\sigma_1\neq ()$ such that 
            $\sigma_1=\bar\sigma_1\circ (l)$, $l\in L$, \textbf{and} 
            $\exists\bar\sigma_2\neq ()$ such that 
            $\sigma_2=\bar\sigma_2\circ (\bot)$ \textbf{and}
            $\sigma_1=\sigma$ \textbf{and} $\sigma_2\sqsubset \sigma$}
               \State $B\leftarrow B\cup \{v\}$
            \EndIf
         \EndFor
         \State \textbf{return} $B$
      \EndFunction
   \end{algorithmic}
\end{algorithm}

\textbf{Routing policies} --
So far we have described how to compose crossing and final pathlets by 
considering all the possible concatenations of available pathlets. Although
this produces the highest possible number of alternative paths, resulting in the
best level of robustness and in the availability of different levels of Quality 
of Service, depending on the topology and on the assignment of 
areas it can be demanding in terms of messages exchanged on the network and of 
pathlets kept at each router. However,
our control plane can also easily accommodate routing policies that influence 
the way in which pathlets are composed and disseminated.
We stress that these policies can be implemented independently for each area: 
that is, the configuration of routing policies on the internal vertices of an 
area may have no impact on the routing information propagated outside that area. 
We believe this is a significant relief for network administrators, who 
do not necessarily need any longer to keep a complete 
knowledge of the network setup and to perform a complex 
planning of configuration changes.

We envision two kinds of
policies: \emph{filters} and \emph{pathlet composition rules}. Filters
can be used to restrict the propagation of pathlets. For example, the 
specification of a
filter on a vertex $u$ can consist of a neighboring vertex $v$ and a
triple $\left<w_1,w_2,\sigma\right>$: when such a filter is applied, $u$ will
avoid propagating to $v$ all those pathlets whose start vertex, end vertex, and
scope stack match the triple.

Pathlet composition
rules can be used to affect the creation of crossing and final pathlets. We 
describe here a few possible pathlet composition rules. As opposed to the 
strategy of considering all the possible concatenations of pathlets,
a border vertex $v$ can create, for each end vertex $w$ of interest, only one 
crossing (or 
final) pathlet that corresponds to an optimal sequence of pathlets 
to that end 
vertex. Several optimality criteria can be pursued. For example, $v$ could 
select the shortest sequence of pathlets by running
Dijkstra's algorithm on the graph resulting by the union of the pathlets it 
knows. We 
highlight that, with
this approach, $v$ can still keep track of possible alternative paths
but does not propagate them as pathlets: in case the shortest sequence of 
pathlets to a
certain vertex $w$ is no longer available (for example because of a failure),
$v$ can transparently switch to an alternative sequence of pathlets leading to
$w$ by just updating the forwarding state and without sending any messages 
outside its area $A_{S(v)}$.
As a variant of this approach, pathlets can be weighted
according to performace indicators (delay, packet loss, jitter) of the network 
portion they traverse: in this case the optimal sequence of pathlets corresponds 
to the one offering the best performance. Alternatively, pathlets 
can be weighted according to their nature of atomic or crossing pathlet: 
assuming that atomic pathlets are assigned weight 
0 and crossing pathlets are assigned weight 1, the optimal pathlet 
tries to avoid transit through areas.
Another pathlet composition rule could accommodate the requirement of an 
administrator that wants to prevent traffic from a specific set $\bar V$ of 
vertices from traversing a specific area $A$. Since detailed routing 
information about the interior of an area is not propagated outside that area, 
it may not be possible to 
establish whether a specific pathlet traverses $A$ or not. Therefore, to 
implement this pathlet composition rule, pathlets could carry an
additional attribute that is a set of \emph{shaded vertices}: crossing pathlets
for $A$ disseminated by the border vertices of $A$ will have the set of shaded
vertices set to $\bar V$; upon receiving a pathlet, a vertex $v$ will
check whether $v$'s identifier appers in the set of shaded vertices and, if so, 
will refrain from using that pathlet for composition or for sending traffic. A
similar mechanism could be implemented to prevent traffic to specific
destinations from traversing $A$: in this case, a set of \emph{shaded
destinations} could be carried in the pathlets instead. Of course, the two
techniques can be combined by using both the set of shaded vertices and the set 
of shaded destinations: in this way, a set of vertices $\bar V$ can be prevented 
from traversing an area $A$ to send traffic to specific destinations.

\vspace{1mm}
{\bf Pathlet dissemination} --
All the created pathlets are disseminated to other vertices in $G$ based on
their scope stacks, as explained in the following. Consider any pathlet
$\pi=\left< \fid,u,v,\sigma,\delta\right>$ and let $\sigma=\bar\sigma\circ (l)$
(by Property~\ref{pro:nonempty-scope-stack}, such $\bar\sigma\neq ()$ and
$l\in L$ must exist). The dissemination of $\pi$ is regulated by the following 
\emph{propagation conditions}. A vertex $w$ can propagate $\pi$ to a
neighboring vertex $n$ either if $n = u$ or if $\pi$'s scope stack does not 
satisfy any of the following conditions:
\begin{enumerate}
   \item $S(w)\Join S(n)\sqsubset\bar\sigma$: restricts propagation of
any pathlets outside the area in which they have been created;
   \item $\sigma\sqsubseteq S(w)\Join S(n)$: prevents propagation of
crossing and final pathlets inside the area of the vertex that created
them;
   \item $\sigma = S(n)\rightarrowtail S(w)$:  prevents $w\notin A$ from
propagating crossing and final pathlets for $A$ inside $A$;
   \item $n = v$: prevents sending to $n$ a pathlet that is useless for $n$.
\end{enumerate}
Conditions 2), 3), and 4) are introduced to prevent the propagation of 
pathlets to vertices that would never use them, thus limiting the amount of 
exchanged 
information during pathlet 
dissemination. Condition 1) can be expressed from the
point of view of a single vertex, leading to the following invariant:
\begin{property}
   All the pathlets received by a vertex $v$ have a scope stack
$\sigma'=\bar\sigma\circ (l)$ such that $\bar\sigma\sqsubseteq S(v)$.
\end{property}

For convenience, given a vertex $w$ that is assigned scope stack 
$S(w)=\sigma_w$, we define $N(w,\sigma_w,\sigma)$ as the set of neighbors of $w$ 
to which $w$ can propagate a pathlet with scope stack $\sigma$ according to the 
propagation conditions and to the routing policies. We assume that 
$N(w,\sigma_w,())=\emptyset$ for any $\sigma_w$.

So far we have mentioned that the propagation conditions regulate the 
propagation of pathlets. However, we will see in Section~\ref{sec:dynamics} that 
other kinds of messages exchanged by our control plane are also propagated 
according to the same conditions.

\vspace{1mm}
{\bf Example of pathlet creation and dissemination} --
To show a complete example of creation and dissemination of pathlets, consider 
again the example in
Fig.~\ref{fig:example} and let $v_6$ host network destination $d$. In the
following we assume that there are no filters applied, 
that the pathlet composition rule is to compose all possible sequences of 
pathlets (although we show only some of them), and that $\fid$s are randomly 
assigned integer numbers, yet
obeying the rules specified in this section.
The atomic pathlet
$\pi_{24,\bot}=\left<1,v_2,v_4,(0\ 1\ \bot),\emptyset\right>$, created by vertex 
$v_2$,
is propagated by $v_2$
to $v_3$ because
$S(v_2)\Join S(v_3) = (0\ 1\ 3)\not\sqsubset (0\ 1)$,
$(0\ 1\ \bot)\not\sqsubseteq (0\ 1\ 3)$,
$(0\ 1\ \bot)\neq S(v_3)\rightarrowtail S(v_2)=(0\ 1\ 3)$,
and
$v_3\neq v_4$;
it is also propagated to $v_1$ for the same reasons.
Instead, $\pi_{24,\bot}$ is not propagated by $v_2$ to $v_6$ because
$S(v_2)\Join S(v_6)=(0)\sqsubset (0\ 1)$ (the first propagation condition
applies), and it is not propagated by
$v_2$ to $v_4$ because the end vertex of $\pi_{24,\bot}$ is $v_4$ itself. For 
similar
reasons, $\pi_{24,\bot}$ is further propagated by $v_3$ to $v_5$, but in turn 
$v_5$ does
not propagate it to $v_7$. Therefore, the visibility
of $\pi_{24,\bot}$ is restricted to vertices inside $A_{(0\ 1)}$. In a similar 
way,
$v_5$ creates the atomic pathlets
$\pi_{53,\bot}=\left<2,v_5,v_3,(0\ 1\ \bot),\emptyset\right>$ and
$\pi_{54,\bot}=\left<3,v_5,v_4,(0\ 1\ \bot),\emptyset\right>$, while
$v_4$ creates the atomic pathlet
$\pi_{46,\bot}=\left<3,v_4,v_6,(0\ \bot),\{d\}\right>$.
The reader can
easily find how these atomic pathlets are propagated. As a border vertex of
$A_{(0\ 1\ 3)}$, $v_3$ will also propagate to $v_5$ the crossing pathlet
$\pi_{32}=\left<1,v_3,v_2,(0\ 1\ 3),\emptyset\right>$ for area
$A_{S(v_3)\rightarrowtail S(v_5)=(0\ 1\ 3)}$. Once pathlets have been
disseminated, $v_5$ has learned about a set of pathlets $\Pi$ and can create a
crossing pathlet for area $A_{S(v_5)\rightarrowtail S(v_7)=(0\ 1)}$ that can be
offered to $v_7$. For example, $v_5$ can pick sequence
$(\pi_{53,\bot}\ \pi_{32}\ \pi_{24,\bot})$ from 
$\chains(\Pi,v_5,v_4,S(v_5)\rightarrowtail
S(v_7))$ and create in its set $\crossing_{v_5}(\Pi,S(v_5)\rightarrowtail
S(v_7))$ the crossing pathlet
$\pi_{54}=\left<1,v_5,v_4,(0\ 1),\emptyset\right>$. Propagation of this pathlet 
by
$v_5$ to $v_4$ is forbidden by the second propagation condition, because
$(0\ 1)\sqsubseteq S(v_5)\Join S(v_4)=(0\ 1)$, and also by the fourth
propagation condition, because $v_4$ is also the end vertex of $\pi_{54}$;
$\pi_{54}$ will however be propagated by
$v_5$ to $v_7$ because
$S(v_5)\Join S(v_7)=(0)\not\sqsubset (0)$,
$(0\ 1)\not\sqsubseteq (0)$,
$(0\ 1)\neq S(v_7)\rightarrowtail S(v_5)=(0\ 2)$,
and
$v_7\neq v_4$.
To provide an alternative path, $v_5$ can create another crossing pathlet
$\pi'_{54}=\left<9,v_5,v_4,(0\ 1),\emptyset\right>$, corresponding to the 
sequence
consisting of the single atomic pathlet $(\pi_{54,\bot})$, and propagated in the 
same 
way as
$\pi_{54}$.
Last, $v_7$ will also create an atomic pathlet
$\pi_{75,\bot}=\left<8,v_7,v_5,(0\ \bot),\emptyset\right>$.
At this point, $v_7$ has two ways to construct a path from itself to
vertex $v_6$, which contains destination $d$: it can concatenate
pathlets $\pi_{75,\bot}$, $\pi_{54}$, and $\pi_{46,\bot}$ or pathlets 
$\pi_{75,\bot}$, $\pi'_{54}$, and
$\pi_{46,\bot}$. The availability of multiple choices supports quick recovery in 
case of
fault and allows $v_7$ to select the pathlet providing the most appropriate
Quality of Service.

\section{A Control Plane for Pathlet Routing: Messages and
Algorithms}\label{sec:dynamics}

We now describe how the dissemination mechanisms illustrated in
Section~\ref{sec:dissemination-basics} are realized in terms of messages
exchanged among vertices and algorithms executed to update routing
information. In this section we also detail how to handle network dynamics, 
including how to deal with topological changes and administrative 
reconfigurations. This actually completes the specification of a control plane 
for pathlet routing.


\subsection{Message Types}
First of all, we detail all the messages that are used by vertices to
disseminate routing information. Each message carries one or more of the
following fields: \texttt{s}: a stack of labels;
\texttt{d}: a set of network destinations; \texttt{p}: a pathlet; \texttt{f}: a
$\fid$; \texttt{a}: a boolean flag (which tells whether a vertex has ``just been 
activated'').
We assume that every message includes an \emph{origin} field \texttt{o} that
specifies the vertex that first originated the message.
Messages can be of the following types, with their fields specified in square
brackets:

\begin{itemize}
   \item \textbf{Hello} [\texttt{s}, \texttt{d}, \texttt{a}] -- Used for neighbor
greetings. It carries the label stack \texttt{s} of the sender vertex, the set 
of network destinations \texttt{d} originated by the sender vertex, and a flag 
\texttt{a} which is set to true when this is the first message sent by a vertex 
since its activation (power-on or reboot). Unlike other message types, 
\textbf{Hello} is only sent to neighbors and is never forwarded. Moreover, in 
order to be able to detect topological variations, it is sent periodically by 
each vertex.
   \item \textbf{Pathlet} [\texttt{p}] -- Used to disseminate a
pathlet \texttt{p}.

\remove{
   \item \textbf{Update} [\texttt{t\_ops}, \texttt{t\_nps}] --
Used to refresh information about already disseminated pathlets, by replacing
every occurrence of the sequence of pathlets in \texttt{t\_ops} with the
sequence \texttt{t\_nps}.
   \item \textbf{Acknowledgment} [\texttt{t\_ops}, \texttt{t\_nps}, \texttt{k},
\texttt{v}] -- Used to acknowledge reception of an
\textbf{Update} message originated by vertex \texttt{v} with the corresponding
\texttt{t\_ops} and \texttt{t\_nps}.
Flag \texttt{k} is set to true if the acknowledged message contained old and
previously known information; to false otherwise.
   \item \textbf{Pathlet\_Mapping} [\texttt{p}, \texttt{ps}] --
Used in some special cases by a vertex $v$ to disseminate a crossing
or final pathlet \texttt{p} as well as the sequence \texttt{ps} of
its component pathlets.
}

   \item \textbf{Withdrawlet} [\texttt{f}, \texttt{s}] -- Used to withdraw the
availability of a pathlet with FID \texttt{f}, scope stack \texttt{s}, and 
start vertex \texttt{o}. We assume that this message can only be 
originated by the vertex that had previously created and disseminated the 
pathlet.
%
   \item \textbf{Withdraw} [\texttt{s}] -- Used
to withdraw the availability of all pathlets having \texttt{s} as scope stack 
and \texttt{o} as start vertex.
\end{itemize}

In order to keep disseminated information consistent in the presence of faults
and reconfigurations, we assume for convenience that all vertices in the network 
have a synchronized clock, and we call $T$ its value at any time.
 Every message type but \textbf{Hello}
 has a \emph{timestamp} field
\texttt{t} that, unless otherwise stated, is set by the sender to the current 
clock $T$ when sending a newly created message; the timestamp is left unchanged 
when a message is just forwarded from a vertex to another. The purpose of the 
timestamp is to let vertices discard outdated messages, which is especially 
important in the presence of faults. In practice, a local counter at each vertex 
can be used in place of the clock value, and its value can be handled in a way 
similar to OSPF sequence numbers (see in particular Section 12.1.6 
of~\cite{rfc1247}).

With the exception of \textbf{Hello}, messages also have a
\emph{source} field \texttt{src} containing the identifier of the vertex that
has sent (or forwarded) the message. This field is also used to
avoid sending the message back to the vertex from which it has been received (a
technique similar to the \emph{split horizon} adopted in commercial routers).
Since the \textbf{Hello} message is never forwarded by any vertices, it contains 
only the origin field.

Given their particular nature, in the following we omit specifying for each
message how the origin, timestamp, and source fields are set, unless we need 
exceptions to their usual assignment.

\subsection{Routing information stored at each vertex}
In our control plane, no vertex has a complete view of all the available routing
paths. However, as a partial representation of the current network status, each
vertex $u\in V$ keeps the following information locally:
\begin{itemize}
   \item For each neighbor $v\in V$ such that $(u,v)\in E$, a label stack
$S_u(v)$ that $u$ currently considers associated with $v$ and a set $D_u(v)$ of
network destinations originated by $v$.
   \item A set $\Pi_u$ of \emph{known pathlets}, consisting of the atomic 
   pathlets created by $u$
and of pathlets that $u$ has received from neighboring vertices. $u$ can
concatenate these pathlets to reach network destinations and, in case it is a 
border vertex, to compose and disseminate crossing and final pathlets.
Each pathlet $\pi\in\Pi_u$ is associated an expiry timer $T_p(\pi)$, that 
specifies how long the pathlet is to be kept in $\Pi_u$ before being removed. 
When a new 
pathlet is created by $u$, its expiry timer is set to the special value 
$T_p(\pi)=\oslash$, meaning that the pathlet never expires.
   \item For every area $A_\sigma$ for which $u$ is a border vertex, a set
$B_u(\sigma)$ of vertices $v\in A_\sigma$, $v\neq u$, that are also border
vertices for $A_\sigma$, and sets $C_u(\sigma)$ and $F_u(\sigma)$ that contain, 
respectively, 
the crossing and final pathlets for area $A_\sigma$ composed by $u$.
   \item A set $H_u$, called \emph{history}, that tracks the most recent piece 
   of information known by $u$ about each pathlet (i.e., not just pathlets in 
   $\Pi_u$). This set consists of t-uples 
   $\left<\fid,v,\sigma,t,\mathit{type}\right>$, where: the $\fid$ and the start 
   vertex $v$ identify a pathlet $\pi$ with scope stack $\sigma$; $t$ is the 
   timestamp of the most recent 
   information that $u$ knows about $\pi$ (it may be the time instant of when 
   $\pi$ has 
   been composed or deleted by $u$, or the timestamp contained in the most 
   recent message received by $u$ about $\pi$); and $\mathit{type}\in\{+,-\}$ 
   determines whether the last known information about $\pi$ is positive ($\pi$ 
   has been composed by $u$ or a \textbf{Pathlet} message has been received 
   about $\pi$) or negative ($\pi$ has been deleted by $u$ or a 
   \textbf{Withdrawlet} or \textbf{Withdraw} message has been received about 
   $\pi$).
   
%
\end{itemize}

The reason why we have introduced an expiry timer $T_p(\pi)$ for each pathlet 
$\pi$ in $\Pi_u$ 
is that we want to prevent indefinite growth of $\Pi_u$. In fact, there may be 
pathlets that can no longer be used by $u$ for concatenations and for which $u$ 
may never receive a 
\textbf{Withdrawlet} or \textbf{Withdraw} message: the expiry timer is used to 
automatically purge such pathlets from 
$\Pi_u$. This situation can occur when a vertex or a link is removed from $G$.
For example, 
consider the network in Fig.~\ref{fig:example} and suppose that $v_2$ composes 
and announces a crossing pathlet
$\pi_{25}=\left<7,v_2,v_5,(0\ 1),\emptyset\right>$ for area $A_{(0\ 1)}$. If 
link $(v_2,v_6)$ fails, $v_6$ has no way to receive a \textbf{Withdrawlet} for 
$\pi_{25}$, because only $v_2$ can originate this message and the propagation 
conditions prevent it from being forwarded inside area $A_{(0\ 1)}$. However, 
$v_6$ can no longer use $\pi_{25}$ for any concatenations and therefore has no 
reason to keep this pathlet in its set $\Pi_{v_6}$: $\pi_{25}$ can indeed be 
automatically removed after timer $T_p(\pi_{25})$ has expired.
The configuration of our control plane therefore requires the specification of 
a timeout value called \emph{pathlet timeout}: this is the value to which the 
expiry timer $T_p(\pi)$ of a pathlet $\pi$ is initialized when $T_p(\pi)$ is 
activated (we will 
see in the following when this activation occurs).

Also the history $H_u$ could grow indefinitely, because an entry is stored and 
kept in $H_u$ even for each deleted or withdrawn pathlet. Therefore, our control 
plane also requires the specification of a \emph{history timeout}: this value 
determines how long negative entries (i.e., with $\mathit{type}=-$) in the 
history $H_u$ of any vertex $u$ are kept before being automatically purged from 
$H_u$. Positive entries (with $\mathit{type}=+$), on the other hand, never 
expire.

In principle, we could completely avoid timeouts and remove pathlets and history 
entries immediately. However, this would significantly increase the number of 
exchanged messages and cause the deletion of pathlets that should instead be 
preserved, even in normal operational conditions. Consider again 
Fig.~\ref{fig:example} and assume there is no pathlet expiry timer. If $v_6$ 
received only pathlet
$\pi_{57,\bot}=\left<6,v_5,v_7,(0\ \bot),\emptyset\right>$ before receiving the 
crossing pathlet $\pi_{25}$, $v_6$ would immediately withdraw $\pi_{57,\bot}$ 
because it cannot use it for concatenations and it may never receive a 
\textbf{Withdrawlet} or \textbf{Withdraw} for that pathlet.
A similar argument applies to the history timer. Look back at 
Fig.~\ref{fig:example} and assume there is no history expiry timer. Note that 
with this assumption negative entries are just not kept in the history, 
actually defeating its purpose. Suppose that, after disseminating an atomic 
pathlet
$\pi_{62,\bot}=\left<1,v_6,v_2,(0\ \bot),\emptyset\right>$ to the whole network,
vertex $v_6$ withdraws this pathlet using a \textbf{Withdrawlet} message (for 
example 
because link $(v_2,v_6)$ has failed). Also suppose that $v_5$ rebooted before 
being able to forward the \textbf{Withdrawlet} to $v_7$: pathlet 
$\pi_{62,\bot}$ would thus be held in set $\Pi_{v_7}$. When $v_5$ becomes again 
active, it 
receives a \textbf{Pathlet} message containing $\pi_{62,\bot}$ from $v_7$, and 
has no way to determine that such information is out-of-date. The only way is to 
propagate pathlet $\pi_{62,\bot}$ to all applicable vertices until it reaches 
$v_6$, which can again withdraw it from the network.

For the sake of clarity, we specify here the strategy with which the history is 
updated when creating or deleting pathlets, and avoid mentioning it again, 
unless there are exceptions to this 
strategy.
Every time a pathlet $\pi=\left<\fid,u,v,\sigma,\delta\right>$ is created by a 
vertex $u$, the history $H_u$ of $u$ is automatically updated with a positive 
entry $\left<\fid,u,\sigma,T,+\right>$, where $T$ denotes the time instant of 
the 
creation. If an entry for the same $\fid$ and start vertex $u$ already existed 
in $H_u$, that entry is replaced by this updated 
version. When pathlet $\pi$ is no longer available (for example because $u$ has 
detected that some of the component pathlets are no longer usable), $H_u$ is 
updated with a negative entry $\left<\fid,u,\sigma,T,-\right>$, where $T$ 
denotes the time instant in which $\pi$ has become unavailable. This entry 
replaces any previously existing entry referring to the same pathlet $\pi$. We 
recall that this negative entry is automatically removed from $H_u$ after the 
history timeout expires.

\subsection{Algorithms to Support Handling of Network Dynamics}
Before actually describing how network dynamics are 
handled, we introduce a few algorithms that vertices execute when they detect a 
change of the locally maintained routing information. In particular, we describe 
the operations performed by a vertex $u$ when its sets 
$\Pi_u$, $C_u$, or $F_u$ are updated. Most of the events that may trigger such 
updates, including topological changes and administrative reconfigurations, can 
be handled based on the algorithms described in this subsection. We discuss in 
detail the application of these algorithms to handle specific events in the 
following subsections.

Suppose the set $\Pi_u$ of currently known pathlets at a vertex $u$ is changed 
and is to be replaced by a new set of known pathlets $\Pi_{\mathit{new}}$. $u$ 
then undertakes the following actions, formalized as procedure 
\textsc{UpdateKnownPathlets}($u$, $\Pi_{\mathit{new}}$) in 
Algorithm~\ref{alg:updatepi}: $u$ sends messages 
to its neighbors to disseminate pathlets that are newly appeared in 
$\Pi_{\mathit{new}}$ (with respect to $\Pi_u$) and withdraw those that are no 
longer in this set. Note 
that, while withdrawn pathlets are immediately removed from $\Pi_u$ at the end 
of the algorithm, the corresponding forwarding state is only cleared by $u$ 
after a timeout $T_f$, in order to allow correct forwarding of data packets 
while the withdraw is propagated on the network (note that the statement at 
line~\ref{algline:clear-forwarding-state-knownpathlets} of 
Algorithm~\ref{alg:updatepi} is non-blocking). Of course some packets could be 
lost if the withdrawn pathlets are physically unavailable. Pathlets 
that existed in $\Pi_u$ but have their scope stack or set of destinations 
updated in 
$\Pi_{\mathit{new}}$ are handled by $u$ in a special way: for each of these 
pathlets 
$u$ disseminates the updated instance of the pathlet to selected neighbors 
(according to the propagation conditions and to the rouing policies set at $u$), 
and withdraws the old instance of the pathlet from other neighbors to which the 
new instance of the pathlet cannot be disseminated. After having updated $\Pi_u$ 
with the contents of $\Pi_{\mathit{new}}$, $u$ checks whether the start vertex 
of each pathlet in $\Pi_u$ is still reachable: it does so by concatenating 
arbitrary pathlets in $\Pi_u$, regardless of their scope stacks. If the start 
vertex of some pathlet $\pi$ is found to be unreachable, or if including $\pi$ 
in 
any sequences of pathlets would always result in a cycle, $u$ can no longer use 
pathlet $\pi$ for composition or traffic forwarding, and it schedules automated 
deletion of the pathlet from $\Pi_u$ by initizializing its expiry 
timer $T_p(\pi)$. For all the other pathlets, the expiry timer is reset, meaning 
that 
they will never expire.
Last, $u$ checks whether the component pathlets of its crossing and final 
pathlets are 
still available and whether new crossing or final pathlets can be composed, and 
updates sets $C_u$ and $F_u$ accordingly. The latter step 
requires further actions, which are detailed in the following procedure.

\begin{algorithm*}
   \caption{Algorithm to update the set $\Pi_u$ of known pathlets at a vertex 
   $u$. The first procedure addresses the case when the label stack of $u$ is 
   contextually changed from $S_{\mathit{old}}$ to $S_{\mathit{new}}$, whereas 
   the second only realizes the update of $\Pi_u$.}
   \label{alg:updatepi}
   \small
   \begin{algorithmic}[1]
      \Procedure{UpdateKnownPathletsAndStack}{$u$, $S_{\mathit{old}}$, 
      $S_{\mathit{new}}$, $\Pi_{\mathit{new}}$}
         \ForAll{$\pi=\left<\fid,v,w,\sigma,\delta\right>\in\Pi_{\mathit{new}}
            \backslash\Pi_u$}
            \label{algline:updatepi-forcycle-1}
            \State \textit{We are considering a pathlet $\pi$ that is not in 
            $\Pi_u$ but is in $\Pi_{\mathit{new}}$ (new pathlet) or the updated 
            instance of a pathlet that is both in $\Pi_u$ and in 
            $\Pi_{\mathit{new}}$}
            \If{$u=v$}
               \State Update $u$'s forwarding state according to the composition 
               of 
               $\pi$
               \State $M\leftarrow$ new \textbf{Pathlet} message
               \State $M.\texttt{p}\leftarrow\pi$
               \ForAll{$n\in N(u,S_{\mathit{new}},\sigma)$}
                  \State Send $M$ to neighbor $n$
               \EndFor
            \EndIf
         \EndFor
         \ForAll{$\pi_{\mathit{old}}=\left<\fid,v,w,\sigma_{\mathit{old}},
            \delta_{\mathit{old}}\right>\in
            \Pi_u\backslash\Pi_{\mathit{new}}$}
            \label{algline:updatepi-forcycle-2}
            \If{$u=v$}
               \State $M\leftarrow$ new \textbf{Withdrawlet} message
               \State $M.\texttt{f}\leftarrow\fid$
               \State $M.\texttt{s}\leftarrow\sigma_{\mathit{old}}$
               \If{$\exists\pi_{\mathit{new}}=\left<\fid,v,w,\sigma_{\mathit{new}},
                  \delta_{\mathit{new}}\right>\in\Pi_{\mathit{new}}$}
                  \State \textit{We are considering a pathlet 
                  $\pi_{\mathit{old}}$ 
                  that is in $\Pi_u$ and has an updated instance 
                  $\pi_{\mathit{new}}$ in $\Pi_{\mathit{new}}$}
                  \ForAll{$n\in 
                  N(u,S_{\mathit{old}},\sigma_{\mathit{old}})\backslash
                     N(u,S_{\mathit{new}},\sigma_{\mathit{new}})$}
                     \State Send $M$ to neighbor $n$
                  \EndFor
               \Else
                  \State \textit{We are considering a pathlet 
                  $\pi_{\mathit{old}}$ 
                  that is in $\Pi_u$ but has been removed in 
                  $\Pi_{\mathit{new}}$}
                  \ForAll{$n\in N(u,S_{\mathit{old}},\sigma_{\mathit{old}})$}
                     \State Send $M$ to neighbor $n$
                  \EndFor
                  \State Clear $\fids_u(\fid)$ and $\nexthop_u(\fid)$ after a 
                  timeout 
                  $T_f$
                  \label{algline:clear-forwarding-state-knownpathlets}
               \EndIf
            \EndIf
         \EndFor
         \State $\Pi_u\leftarrow \Pi_{\mathit{new}}$
         \ForAll{$\pi=\left<\fid,v,w,\sigma,\delta\right>\in\Pi_u$}
            \If{$\chains(\Pi_u,u,v,())=\emptyset$ or any concatenation of a 
            pathlet in $\chains(\Pi_u,u,v,())$ with pathlet $\pi$ has a cycle}
               \State $T_p(\pi)\leftarrow$ value of the pathlet timeout parameter
            \Else
               \State $T_p(\pi)\leftarrow\oslash$
            \EndIf
         \EndFor
         \State \Call{UpdateComposedPathletsAndStack}{$u$, $S_{\mathit{old}}$, 
            $S_{\mathit{new}}$, $\Pi_u$}
      \EndProcedure
      \Procedure{UpdateKnownPathlets}{$u$, $\Pi_{\mathit{new}}$}
         \State \Call{UpdateKnownPathletsAndStack}{$u$, $S(u)$, $S(u)$, 
         $\Pi_{\mathit{new}}$}
      \EndProcedure
   \end{algorithmic}
\end{algorithm*}

When set $\Pi_u$ is replaced by $\Pi_{\mathit{new}}$, vertex $u$ must also check 
whether the component 
pathlets for its crossing and final pathlets are still available in 
$\Pi_{\mathit{new}}$ and 
whether there are new crossing and final pathlets that $u$ should compose due to 
newly appered pathlets in $\Pi_{\mathit{new}}$. 
Function \textsc{IsPathletComposable}($u$, $\pi$, $\Pi$, $E$) in 
Algorithm~\ref{alg:is-pathlet-composable} can be used to establish whether a 
certain pathlet $\pi$ can (still) be composed by $u$ based on the set $\Pi$ of 
known pathlets at $u$ and on a set $E$ of admissible end vertices for $\pi$ (the 
check performed by this function actually reflects 
the 
mechanism for the construction of set $\crossing_u$ as 
explained in Section~\ref{sec:dissemination-basics}).
The composition (or deletion) of 
crossing and final pathlets also depends on the areas for which $u$ is a border 
vertex and on the knowledge of other border vertices.
All the operations that $u$ is supposed to 
execute to update its crossing and final pathlets are therefore formalized as 
procedure 
\begin{algorithm}[t!]
   \caption{Algorithm to check whether a pathlet $\pi$ can (still) be composed 
      by a vertex $u$ given a set $\Pi$ of known pathlets and a set $E$ of 
      admissible end vertices for $\pi$.}
   \label{alg:is-pathlet-composable}
   \small
   \begin{algorithmic}[1]
      \Function{IsPathletComposable}{$u$, $\pi$, $\Pi$, $E$}
         \State Let $\pi=\left<\fid,u,v,\sigma,\delta\right>$
         \If{$v\in E$ \textbf{ and } $\exists (\pi_1\ \pi_2\ \dots\ 
            \pi_n)\in\chains(\Pi,u,v,\sigma)$ 
            such that $\pi_i=\left<\fid_i,u_i,v_i,\sigma_i,\delta_i\right>$,
            $i=1,\dots,n$ \textbf{and} $\fids_u(\fid)=(\fid_2\ \fid_3\ \dots\ 
            \fid_n)$ 
            \textbf{and} $\nexthop_u(\fid)=u_2$ \textbf{and} the pathlet 
            composition rules allow composition of $\pi$}
            \State \textbf{return} \texttt{True}
         \Else
            \State \textbf{return} \texttt{False}
         \EndIf
      \EndFunction
   \end{algorithmic}
\end{algorithm}
\textsc{UpdateComposedPathlets}($u$, $\Pi_{\mathit{new}}$) in 
Algorithm~\ref{alg:update-crossing-final}: 
$u$ considers pathlets that it can no longer compose ($C_{\mathit{old}}$) 
because it is no longer a border vertex for some area or because some of the 
component pathlets are no longer available in $\Pi_{\mathit{new}}$; $u$ may also 
compose new crossing and final pathlets ($C_{\mathit{new}}$) because it has 
become a 
border vertex for some area or because there are new possible compositions of 
pathlets in $\Pi_{\mathit{new}}$. If possible, $u$ attempts to transparently 
replace pathlets in $C_{\mathit{old}}$ with newly composed pathlets from 
$C_{\mathit{new}}$ by just updating its forwarding 
state and without sending any messages; if this is not possible, $u$ 
withdraws the no longer available 
pathlets from those neighbors to which they had been disseminated and clears the 
forwarding state for these pathlets after a timeout $T_f$ (the statement at 
line~\ref{algline:clear-forwarding-state-composedpathlets} of 
Algorithm~\ref{alg:update-crossing-final} is non-blocking). Last, $u$ 
disseminates to selected neighbors (according to the propagation conditions and 
the routing policies) all those newly composed pathlets in $C_{\mathit{new}}$ 
that were not used as a replacement for pathlets in $C_{\mathit{old}}$. Pathlet 
composition at line~\ref{algline:update-crossing-final-nocrossing} of 
Algorithm~\ref{alg:update-crossing-final} follows the same mechanism as for set 
$\crossing$: we did not use set $\crossing(\Pi_{\mathit{new}},\sigma)$ here 
because the 
$\fid$s of already existing pathlets in $C_u(\sigma)$ must be retained.

\begin{algorithm*}
   \caption{Algorithm to update the sets of crossing and final pathlets composed 
   by a vertex $u$. The first procedure considers the case when the label stack 
   of $u$ is contextually changed from $S_{\mathit{old}}$ to $S_{\mathit{new}}$, 
   whereas the second only realizes the update of crossing and final pathlets.}
   \label{alg:update-crossing-final}
   \small
   \begin{algorithmic}[1]
      \Procedure{UpdateComposedPathletsAndStack}{$u$, $S_{\mathit{old}}$, 
      $S_{\mathit{new}}$, $\Pi_{\mathit{new}}$}
         \ForAll{area $A_\sigma$}
            \State $C_{\mathit{new}}\leftarrow\emptyset$
            \State $C_{\mathit{old}}\leftarrow\emptyset$
            \If{$u$ is a border vertex for $A_\sigma$}
               \State $B_u(\sigma)\leftarrow$ 
               \Call{DiscoverBorderVertices}{$u$, $\sigma$, $\Pi_{\mathit{new}}$}
               \If{$C_u(\sigma)=\emptyset$}
                  \State \textit{Vertex $u$ has become a border vertex for 
                  $A_\sigma$ or has not yet composed any pathlets for that area; 
                  pathlets in set 
                  $\crossing_u(\Pi_{\mathit{new}},\sigma)$ below have end 
                  vertices 
                  in $B_u(\sigma)$}
                  \State 
                  $C_{\mathit{new}}\leftarrow\crossing_u(\Pi_{\mathit{new}},\sigma)$
               \Else
                  \State \textit{Vertex $u$ continues to be a border vertex for 
                     $A_\sigma$, but it has to refresh available crossing 
                     pathlets 
                     according to the contents of $\Pi_{\mathit{new}}$}
                  \State $C_{\mathit{new}}\leftarrow$ new crossing pathlets not 
                  in
                  $C_u(\sigma)$, that $u$ can compose towards vertices in 
                  $B_u(\sigma)$ using pathlets in
                  $\Pi_{\mathit{new}}$ and according to the 
                  pathlet composition rules
                  \label{algline:update-crossing-final-nocrossing}
                  \State $C_{\mathit{old}}\leftarrow \{\pi | \pi\in C_u(\sigma) 
                  \textbf{ and not } 
                  \textsc{IsPathletComposable}(u,\pi,\Pi_{\mathit{new}},B_u(\sigma))\}$
               \EndIf
            \ElsIf{$C_u(\sigma)\neq\emptyset$}
               \State \textit{Vertex $u$ was a border vertex for $A_\sigma$ but 
               is no longer}
               \State $C_{\mathit{old}}\leftarrow C_u(\sigma)$
            \EndIf
            \State Update $u$'s forwarding state for any pathlet in 
            $C_{\mathit{new}}$
            \If{$C_{\mathit{old}}=C_u(\sigma)$}
               \State \textit{All the crossing pathlets have been removed: this 
               piece of information can be propagated with a single Withdraw 
               message}
               \State $M\leftarrow$ a new \textbf{Withdraw} message
               \State $M.\texttt{s}\leftarrow\sigma$
               \ForAll{$n\in N(u,S_{\mathit{old}},\sigma)$}
                  \State Send $M$ to neighbor $n$
               \EndFor
            \Else               
               \ForAll{$\pi_{\mathit{old}}=\left<\fid_{\mathit{old}},v,w,
                  \sigma,\delta\right>\in C_{\mathit{old}}$}
                  \If{$\exists\pi_{\mathit{new}}=\left<\fid_{\mathit{new}},v,w,
                     \sigma,\delta\right>\in C_{\mathit{new}}\backslash
                     C_{\mathit{old}}$}
                     \State \textit{Use an alternative pathlet 
                     $\pi_{\mathit{new}}$ 
                     to transparently replace a no longer available pathlet 
                     $\pi_{\mathit{old}}$ by only updating $u$'s forwarding 
                     state}
                     \State $\fids_u(\fid_{\mathit{old}})\leftarrow
                     \fids_u(\fid_{\mathit{new}})$
                     \State $\nexthop_u(\fid_{\mathit{old}})\leftarrow
                     \nexthop_u(\fid_{\mathit{new}})$
                     \State $C_{\mathit{new}}\leftarrow 
                     (C_{\mathit{new}}\backslash\{\pi_{\mathit{new}}\})\cup
                     \{\pi_{\mathit{old}}\}$
                  \Else
                     \State $M\leftarrow$ a new \textbf{Withdrawlet} message
                     \State $M.\texttt{f}\leftarrow \fid_{\mathit{old}}$
                     \State $M.\texttt{s}\leftarrow \sigma$
                     \ForAll{$n\in N(u,S_{\mathit{old}},\sigma)$}
                        \State Send $M$ to neighbor $n$
                     \EndFor
                     \State Clear $\fids_u(\fid_{\mathit{old}})$ and 
                     $\nexthop_u(\fid_{\mathit{old}})$ after a 
                     timeout $T_f$
                     \label{algline:clear-forwarding-state-composedpathlets}
                  \EndIf
               \EndFor
               
               \ForAll{$\pi_{\mathit{new}}=\left<\fid_{\mathit{new}},v,w,
                  \sigma,\delta\right>\in C_{\mathit{new}}\backslash 
                  C_{\mathit{old}}$}
                  \State $M\leftarrow$ a new \textbf{Pathlet} message
                  \State $M.\texttt{p}\leftarrow\pi$
                  \ForAll{$n\in N(u,S_{\mathit{new}},\sigma)$}
                     \State Send $M$ to neighbor $n$
                  \EndFor
               \EndFor
            \EndIf
            \State $C_u(\sigma)\leftarrow (C_u(\sigma)\backslash 
            C_{\mathit{old}})\cup 
            C_{\mathit{new}}$
         \EndFor
         \State Repeat the same steps replacing set $C_u(\sigma)$ with 
         $F_u(\sigma)$, set $B_u(\sigma)$ with $A_\sigma\cap D$, set 
         $\crossing(\Pi_{\mathit{new}},\sigma)$ with 
         $\final(\Pi_{\mathit{new}},\sigma)$, and ``crossing pathlets'' with 
         ``final pathlets''
      \EndProcedure
      \Procedure{UpdateComposedPathlets}{$u$, $\Pi_{\mathit{new}}$}
         \State \Call{UpdateComposedPathletsAndStack}{$u$, $S(u)$, $S(u)$, 
         $\Pi_{\mathit{new}}$}
      \EndProcedure
   \end{algorithmic}
\end{algorithm*}

\subsection{Handling Topological Variations and Configuration 
Changes}\label{sec:stack-change}

As soon as a vertex $u$ becomes active on the network, it sends a \textbf{Hello} 
message $M$ to all its neighbors, with $M.\texttt{s}$ set to its label stack 
$S(u)$, $M.\texttt{d}$ set to the available destinations at $u$ (if any), and 
$M.\texttt{a}=\texttt{True}$.
With this simple neighbor greeting mechanism, each vertex can learn about its 
neighborhood. Once a vertex has collected this information, it
starts creating and disseminating atomic pathlets as explained in 
Section~\ref{sec:dissemination-basics}. Although this reasonably summarizes the 
behavior of a vertex that has just appeared on the network, ``becoming active'' 
is just one of the possible topological variations that graph $G$ may undergo
during network operation. Moreover, our control plane must also support 
administrative configuration changes that can occur while the network is running.

In our model, most topological variations and configuration changes can be 
represented as a change of label stacks, in the following way: 
addition of a link $(u,v)$ is modeled by the assignment of value $S(v)$ to label 
stack $S_u(v)$ and of value $S(u)$ to label stack $S_v(u)$; removal of a link 
$(u,v)$ is 
modeled as a change of label stacks $S_u(v)$ and $S_v(u)$ to the empty stack 
$()$; addition and removal of a vertex are modeled as a simultaneous addition 
or removal of all its incident edges; an administrative 
configuration change that modifies the label stack $S(v)$ assigned to a vertex 
$v$ is modeled as an update of stacks $S_w(v)$ of all the neighbors $w$ of $v$.
To complete the picture of possible reconfigurations, we assume that a change in 
the routing policies of a vertex 
causes a reboot of that vertex (this assumption can be removed, but then each 
vertex has to keep track of the pathlets it has propagated): we therefore do not 
discuss this kind of configuration change further.
For these reasons, we can handle all relevant network dynamics by defining a 
generic 
algorithm to deal with a change of the known label stack of a vertex. We will 
see in the rest of this section that this algorithm is designed to limit the 
propagation of 
the effect of a network change: in fact, only those pathlets that involve 
vertices affected by the change are disseminated as a consequence of the change. 
Moreover, we enforce mechanisms to transparently replace a pathlet that is no 
longer available without the need to disseminate any information to the rest of 
the network.


In principle, we could define an algorithm for ``push'' and
``pop'' primitives on the stack of $v$ and consider a generic stack
change as consisting of a suitable sequence of pop operations followed by
push operations. However, this choice has two drawbacks: first of all, care
should be taken in order to avoid that push and pop operations triggered by
different network events are mixed up, resulting in inconsistent assignments
of label stacks; second, implementing a stack change as a sequence of push and
pop operations results in more messages being exchanged. As an
example, consider again the network in Fig.~\ref{fig:example} and suppose
that vertex $v_2$ has its stack administratively changed
from $S(v_2)=(0\ 1\ 3)$ to $S(v_2)=(0\ 2\ 1)$: if this event were implemented
with push and pop primitives, $v_2$ would also be assigned the intermediate
stack $(0\ 1)$, which would make $v_1$ a border vertex for area $A_{(0\ 1\ 3)}$
and cause $v_1$ to disseminate appropriate crossing (and final) pathlets for
that area. Instead, in the final state in which $S(v_2)=(0\ 2\ 1)$, $v_1$ is not
supposed to disseminate these pathlets, because it a border vertex only for
area $A_{(0\ 1)}$.
We therefore consider the stack change as an atomic operation in the
following.


\vspace{1mm}
{\bf Stack change} --
We now describe the operations performed by a vertex when its label stack is 
administratively changed, for example because the vertex is moved to a different 
area. Despite being also modeled as a stack change, this does not include the
case when the vertex fails, because of course it would not be able to undertake
any actions: this case is handled just as if
neighbors of the failed vertex received a \textbf{Hello} message from that
vertex with \texttt{s} set to $()$, and is therefore discussed later on.

Consider a vertex $u\in V$ and suppose its label stack $S(u)$ is changed at a
certain time instant from $\stackold$ to $\stacknew$. As a consequence of this 
change, 
some pathlets may be created or deleted by $u$, or have their scope stack 
changed.
The following steps describe how pathlets are modified by $u$ and which messages 
are generated by $u$ to disseminate this information.
\begin{enumerate}
\item $u$ informs all its neighbors that its label stack has changed. 
To this purpose, $u$ sends to each of its neighbors a $\textbf{Hello}$ message 
$M$ with $M.\texttt{s}=\stacknew$, $M.\texttt{d}$ set according to the network
destinations available at $u$, and $M.\texttt{a}=\mathit{false}$.
\item $u$ considers the atomic pathlets towards its neighbors.
Since the stack change may influence the scope stack of some of these 
pathlets, $u$ may have to update and disseminate them to a relevant subset of 
neighbors. Observe that, because of the propagation 
conditions and of the routing policies, an updated atomic pathlet may not be 
propagated to the same 
neighbors to 
which it was propagated before the stack change. Hence, $u$ will send to some 
neighbors \textbf{Pathlet} messages that announce or update some atomic 
pathlets, and to other neighbors \textbf{Withdrawlet} messages that withdraw 
atomic pathlets that should no longer be visible.\\
Formally, for each neighbor $v$ of $u$, 
if $(\stackold\Join S_u(v)) \neq 
(\stacknew\Join S_u(v))$, $u$ searches $\Pi_u$ for an atomic pathlet 
$\pi_{\mathit{old}}=\left<\fid,u,v,\sigma_{\mathit{old}},\delta\right>$ from $u$ 
to $v$. 
Such pathlet must exist, because at least it has been created immediately after 
$u$ has received a \textbf{Hello} message from $v$.
Let $\pi_{\mathit{new}}=\left<\fid,u,v,\sigma_{\mathit{new}},\delta\right>$, with
$\sigma_{\mathit{new}}=(\stacknew\Join S_u(v)) \circ (\bot)$.
Then, $u$ executes procedure \textsc{UpdateKnownPathletsAndStack}($u$, 
$S_{\mathit{old}}$, $S_{\mathit{new}}$, 
$(\Pi_u\backslash\{\pi_{\mathit{old}}\})\cup\{\pi_{\mathit{new}}\}$) from 
Algorithm~\ref{alg:updatepi}.

\item $u$ considers the areas to which its neighbors belong and updates its role 
of border vertex: if $u$ 
is no longer a border vertex for some areas after the stack change, it must
delete all crossing 
and final pathlets for these areas and withdraw them to relevant neighbors.
Conversely, if after the stack 
change $u$ becomes a border vertex for some areas, it must create crossing and
final pathlets for these areas and
disseminate them to the relevant neighbors, according to the
propagation conditions and to the routing policies. For the areas for which $u$ 
continues to be a border vertex, it must check whether the pathlets that make up 
its crossing and final pathlets are still available, or whether new compositions 
are possible, and disseminate the corresponding information. To realize these 
operations, $u$ executes procedure \textsc{UpdateComposedPathletsAndStack}($u$, 
$S_{\mathit{old}}$, $S_{\mathit{new}}$, 
$(\Pi_u\backslash\{\pi_{\mathit{old}}\})\cup\{\pi_{\mathit{new}}\}$), which is 
invoked within \textsc{UpdateKnownPathletsAndStack}.
%
\end{enumerate}

\vspace{1mm}
{\bf Update of network destinations} --
When an administrative configuration change modifies the set of network 
destinations available at a certain vertex $u$, all vertices that store a 
pathlet with $u$ as an end vertex must have this pathlet updated with the new 
available destinations. To achieve this, $u$ peforms only step 1) of the stack 
change: this is enough to propagate the updated information.
In fact, as shown in the next subsection, when a vertex $v$ receives a 
$\textbf{Hello}$ or a $\textbf{Pathlet}$ message that carries already known 
information but for the set of destinations, $v$ updates the pathlets it 
stores locally and forwards the updated information 
to its neighbors, according to the propagation conditions and to the routing 
policies.

\subsection{Message Handling}

In the previous subsections we have described the actions performed by a vertex 
when it detects a topological change or it undergoes a configuration change. 
Therefore, to complete the specification of the control plane, we need to 
specify the behavior of a vertex when it receives any of the messages introduced 
in this section.
Assume that vertex $u$ receives a message $M$. The actions performed by $u$ 
depend on the type of message $M$, and are detailed in the following.


\vspace{1mm}
\textbf{Receipt of a Hello Message} --
When a vertex $u$ receives a $\textbf{Hello}$ message $M$ from a neighbor 
$M.\texttt{o}$, it performs several actions.

First of all, $u$ updates its knowledge about vertex $M.\texttt{o}$ by setting 
$S_u(M.\texttt{o})=M.\texttt{s}$ and $D_u(M.\texttt{o})=M.\texttt{d}$.

After that, $u$ checks whether 
$M.\texttt{a}=\texttt{True}$, which means that this is the first \textbf{Hello} 
message sent by $M.\texttt{o}$ since its activation. If this is the case, 
vertex $M.\texttt{o}$ needs to learn about all the currently available 
pathlets. For this reason, $u$ sends to $M.\texttt{o}$ all the information it 
currently knows, and in particular: for every pathlet 
$\pi=\left<\fid,v,w,\sigma,\delta\right>$ in any of the sets $\Pi_u$, $C_u$, and 
$F_u$ kept by $u$ such that $M.\texttt{o}\in N(u,S(u),\sigma)$, $u$ sends to 
$M.\texttt{o}$ a \textbf{Pathlet} message $M_P$ with $M_P.\texttt{p}=\pi$, 
$M_P.\texttt{o}=v$, and $M_P.\texttt{t}=t$, where $t$ is taken from entry 
$\left<\fid,v,\sigma,t,+\right>$ in history $H_u$ (note that such an entry must 
exist for every pathlet learned or created by $u$). Moreover, for every entry 
$\left<\fid,v,\sigma,t,-\right>$ in history $H_u$ such that $M.\texttt{o}\in 
N(u,S(u),\sigma)$, $u$ sends to $M.\texttt{o}$ a \textbf{Withdrawlet} message 
$M_W$ with $M_W.\texttt{f}=\fid$, $M_W.\texttt{s}=\sigma$, $M_W.\texttt{o}=v$, 
and $M_W.\texttt{t}=t$. Observe that, in sending these messages, $u$ preserves 
the origin vertex and timestamp of the originally learned information, as 
specified in the history.

At this point, $u$ creates or updates pathlets as required, based on the newly 
learned information about its neighbor $M.\texttt{o}$.
As a first step, $u$ creates an atomic pathlet towards $M.\texttt{o}$, or 
updates it if it already exists in $\Pi_u$. Since this action may change the 
contents of $\Pi_u$, several crossing and final pathlets may also need to be 
created or deleted, based on the availability of their component pathlets. 
Moreover, after $u$ has learned about the label stack of $M.\texttt{o}$, it 
can detect that its status of border vertex for some areas has changed (it may 
become border vertex for some new areas and cease being border vertex for other 
areas), and this also requires updating crossing and final pathlets.

More formally, let 
$\pi_{\mathit{new}}=\left<\fid_{\mathit{new}},u,v,\sigma_{\mathit{new}},
\delta_{\mathit{new}}\right>$ be a new atomic pathlet from $u$ to 
$M.\texttt{o}$, 
with $\fid_{\mathit{new}}$ chosen to be unique at 
$u$, $v=M.\texttt{o}$, 
$\sigma_{\mathit{new}}=(S(u)\Join 
S_u(v))\circ (\bot)$, and $\delta_{\mathit{new}}=D_u(v)$. If a pathlet 
$\pi_{\mathit{old}}=\left<\fid_{\mathit{old}},u,v,\sigma_{\mathit{old}},
\delta_{\mathit{old}}\right>$ exists in $\Pi_u$, then let 
$\fid_{\mathit{new}}=\fid_{\mathit{old}}$ (that is, the old pathlet is updated) 
and $\Pi_{\mathit{old}}=\{\pi_{\mathit{old}}\}$; otherwise, let 
$\Pi_{\mathit{old}}=\emptyset$.
To realize all the required pathlet update operations, including those of 
crossing and final pathlets, and disseminate the updated information, $u$ 
executes procedure \textsc{UpdateKnownPathlets}($u$, 
$(\Pi_u\backslash\Pi_{\mathit{old}})\cup\{\pi_{\mathit{new}}\}$) from 
Algorithm~\ref{alg:updatepi}.

Note that, even if vertex $M.\texttt{o}$ has sent an updated label stack, for 
example due to a stack change, it may be the case that no pathlets are updated 
by $u$ and no messages are sent by $u$. In fact, if the atomic pathlet 
$\pi_{\mathit{old}}$ from $u$ to $M.\texttt{o}$ already existed 
in $\Pi_u$ and its scope stack $\sigma_{\mathit{old}}$ and set of destinations 
$\delta_{\mathit{old}}$ are unchanged in $\pi_{\mathit{new}}$ (which, for the 
scope stack, only requires that $S(u)\Join S_u(v)$ is unchanged), $u$ does not 
perform any actions: this is visible in Algorithm~\ref{alg:updatepi} because the 
two \textbf{for} cycles at lines~\ref{algline:updatepi-forcycle-1} 
and~\ref{algline:updatepi-forcycle-2} execute no iterations since 
$\Pi_{\mathit{new}}=\Pi_u$; moreover, if $S(u)\Join 
S_u(v)$ is unchanged, $u$ cannot change either the areas for which it is a 
border vertex or any of the sets $B_u$, and this causes sets $C_{\mathit{old}}$ 
and $C_{\mathit{new}}$ in Algorithm~\ref{alg:update-crossing-final} to be empty, 
resulting in no actions being performed even during the execution of that 
algorithm.

Last, $u$ checks whether the set of available destinations at its neighbor 
$M.\texttt{o}$ has changed. In particular, for each pathlet 
$\pi_{\mathit{old}}=\left<\fid,v,w,\sigma,\delta_{\mathit{old}}\right>$ in 
$\Pi_u$ or in any of the sets $F_u$ such that $w=M.\texttt{o}$ and 
$\delta_{\mathit{old}}\neq D_u(w)$, $u$ sends to all its neighbors in 
$N(u,S(u),\sigma)$ a \textbf{Pathlet} message $M$ with 
$M.\texttt{p}=\pi_{\mathit{new}}$, where 
$\pi_{\mathit{new}}=\left<\fid,v,w,\sigma,D_u(w)\right>$.

\vspace{1mm}
\textbf{Receipt of a Pathlet Message} --
Upon receiving a \textbf{Pathlet} message carrying a pathlet 
$\pi_{\mathit{msg}}=M.\texttt{p}$, a vertex $u$ first of all checks the 
freshness of the 
information contained in that message: if the information contained in the 
message is older than the information that $u$ currently has about 
$\pi_{\mathit{msg}}$, $u$ must send back a message with the updated 
information; otherwise, $u$ accepts the fresher information and updates its 
pathlets and history accordingly.

In particular, let $\pi_{\mathit{msg}}=\left<\fid,v,w,\sigma_{\mathit{msg}},
\delta_{\mathit{msg}}\right>$. If $u$ is the originator of $\pi_{\mathit{msg}}$, 
namely $u=v$, then the information 
known by $u$ about $\pi_{\mathit{msg}}$ is to be considered always fresher, and 
the message can never carry updated information. If $u$ is not the originator of 
$\pi_{\mathit{msg}}$, the freshness of message $M$ is determined by comparing 
the message timestamp $M.\texttt{t}$ with the timestamp of the most recent 
information that $u$ keeps about $\pi_{\mathit{msg}}$ in its history $H_u$. In 
all the cases in which the information received in message $M$ is outdated, $u$ 
replies with a message containing the updated information.
Function \textsc{IsPathletMessageFresher}($u$, $M$) in 
Algorithm~\ref{alg:is-pathlet-message-fresher} realizes this freshness check and 
returns \texttt{True} only when message $M$ carries updated information. This 
function also sends updated information back to $M.\texttt{src}$, forwards the 
received message to relevant neighbors, and updates the history $H_u$ as 
required.

\begin{algorithm*}
   \caption{Algorithm to determine whether a \textbf{Pathlet} message $M$ 
      carries updated information about a pathlet: the function returns 
      \texttt{True} only in this case. It also handles message forwarding and 
      history update.}
   \label{alg:is-pathlet-message-fresher}
   \small
   \begin{algorithmic}[1]
      \Function{IsPathletMessageFresher}{$u$, $M$}
         \State $\pi_{\mathit{msg}}\leftarrow M.\texttt{p}$
         \State Let $\pi_{\mathit{msg}}=\left<\fid,v,w,\sigma_{\mathit{msg}},
         \delta_{\mathit{msg}}\right>$
         \If{$u=v$}
            \State \textit{$u$ is the originator of pathlet $\pi_{\mathit{msg}}$}
            \If{there is no pathlet identified by $\fid$ and with start vertex 
            $u$ in $\Pi_u$ or in any 
               of the sets $C_u$ and $F_u$}
               \State $M_W\leftarrow$ new \textbf{Withdrawlet} message
               \State $M_W.\texttt{f}\leftarrow\fid$
               \State $M_W.\texttt{s}\leftarrow\sigma_{\mathit{msg}}$
               \State Send $M_W$ to neighbor $M.\texttt{src}$
            \Else
               \State $\pi_{\mathit{cur}}\leftarrow$ pathlet identified by 
               $\fid$ and with start vertex $u$ that is known at $u$
               \If{$\pi_{\mathit{msg}}\neq\pi_{\mathit{cur}}$}
                  \State $M_P\leftarrow$ new \textbf{Pathlet} message
                  \State $M_P.\texttt{p}\leftarrow\pi_{\mathit{cur}}$
                  \State Send $M_P$ to neighbor $M.\texttt{src}$
               \EndIf
            \EndIf
            \State \textbf{return} \texttt{False}
         \Else
            \If{$\exists\left<\fid,v,\sigma,t,
               \mathit{type}\right>$ in $H_u$}
               \If{$t<M.\texttt{t}$}
                  \State Replace $\left<\fid,v,\sigma_{\mathit{msg}},t,
                  \mathit{type}\right>$ in $H_u$ with 
                  $\left<\fid,v,\sigma_{\mathit{msg}},M.\texttt{t},+\right>$
                  \ForAll{$n\in N(u,S(u),\sigma_{\mathit{msg}})\backslash 
                     \{M.\texttt{src}\}$}
                     \State Send $M$ to neighbor $n$
                  \EndFor
                  \State \textbf{return} \texttt{True}
               \Else
                  \If{$\mathit{type}=+$}
                     \State $\pi_{\mathit{cur}}\leftarrow$ pathlet in $\Pi_u$ 
                     identified by 
                     $\fid$ and with start vertex $v$
                     \State $M_P\leftarrow$ new \textbf{Pathlet} message
                     \State $M_P.\texttt{p}\leftarrow\pi_{\mathit{cur}}$
                     \State $M_P.\texttt{t}\leftarrow t$
                     \State Send $M_P$ to neighbor $M.\texttt{src}$
                  \Else
                     \State $M_W\leftarrow$ new \textbf{Withdrawlet} message
                     \State $M_W.\texttt{f}\leftarrow\fid$
                     \State $M_W.\texttt{s}\leftarrow\sigma$
                     \State $M_W.\texttt{t}\leftarrow t$
                     \State Send $M_W$ to neighbor $M.\texttt{src}$
                  \EndIf
                  \State \textbf{return} \texttt{False}
               \EndIf
            \Else
               \State Add 
               $\left<\fid,v,\sigma_{\mathit{msg}},M.\texttt{t},+\right>$ to 
               $H_u$
               \ForAll{$n\in N(u,S(u),\sigma_{\mathit{msg}})\backslash 
                  \{M.\texttt{src}\}$}
                  \State Send $M$ to neighbor $n$
               \EndFor
               \State \textbf{return} \texttt{True}
            \EndIf
         \EndIf
      \EndFunction
   \end{algorithmic}
\end{algorithm*}

$u$ therefore executes function \textsc{IsPathletMessageFresher}($u$, $M$): if 
it returns $\texttt{False}$, the handling of $M$ by $u$ is finished, because the 
message does not carry any useful information (and function 
\textsc{IsPathletMessageFresher} already takes care of forwarding the 
\textbf{Pathlet} message as appropriate).
Otherwise, $u$ looks in its set $\Pi_u$ for a pathlet 
$\pi_{\mathit{old}}=\left<\fid,v,w,\sigma_{\mathit{old}},
\delta_{\mathit{old}}\right>$. If this pathlet exists, $u$ sets 
$\Pi_{\mathit{old}}=\{\pi_{\mathit{old}}\}$, otherwise $u$ sets 
$\Pi_{\mathit{old}}=\emptyset$. At this point, $u$ updates its sets of known 
pathlets, crossing pathlets, and final pathlets, as well as its status of border 
vertex and sets $B_u$ of other border vertices, and disseminates updated 
information to its neighbors. All these tasks are accomplished 
by $u$ by executing procedure \textsc{UpdateKnownPathlets}($u$, 
$(\Pi_u\backslash\Pi_{\mathit{old}})\cup \{\pi_{\mathit{msg}}\}$).

\vspace{1mm}
\textbf{Receipt of a Withdrawlet Message} --
Handling of a \textbf{Withdrawlet} message $M$ received by a vertex $u$ is much 
similar to that of a \textbf{Pathlet} message.
First of all, $u$ checks the freshness of the information carried by $M$ by 
invoking function \textsc{IsWithdrawletMessageFresher}($u$, $M$) in 
Algorithm~\ref{alg:is-withdrawlet-message-fresher}: if this function returns 
$\texttt{False}$, then handling of message $M$ is completed.

Otherwise, $u$ searches $\Pi_u$ for pathlet 
$\pi_{\mathit{old}}=\left<M.\texttt{f},M.\texttt{o},w,M.\texttt{s},\delta\right>$.
 If this pathlet exists in $\Pi_u$, $u$ updates pathlets and disseminates 
 information by executing procedure 
 \textsc{UpdateKnownPathlets}($u$, $\Pi_u\backslash\{\pi_{\mathit{old}}\}$); 
 otherwise $u$ undertakes no further actions, because there is no pathlet to be 
 withdrawn. Note that, regardless of whether $\pi_{\mathit{old}}$ exists in 
 $\Pi_u$, function \textsc{IsWithdrawletMessageFresher} already takes care of 
 appropriately forwarding the \textbf{Withdrawlet} message.

\begin{algorithm*}
   \caption{Algorithm to determine whether a \textbf{Withdrawlet} message $M$ 
      carries updated information about a pathlet: the function returns 
      \texttt{True} only in this case. It also handles message forwarding and 
      history update.}
   \label{alg:is-withdrawlet-message-fresher}
   \small
   \begin{algorithmic}[1]
      \Function{IsWithdrawletMessageFresher}{$u$, $M$}
         \If{$\exists\left<M.\texttt{f},M.\texttt{o},\sigma,t,
            \mathit{type}\right>$ in $H_u$}
            \If{$t<M.\texttt{t}$}
               \State Replace $\left<M.\texttt{f},M.\texttt{o},\sigma,t,
               \mathit{type}\right>$ in $H_u$ with 
               $\left<M.\texttt{f},M.\texttt{o},M.\texttt{s},M.\texttt{t},-\right>$
               \ForAll{$n\in N(u,S(u),M.\texttt{s})\backslash 
                  \{M.\texttt{src}\}$}
                  \State Send $M$ to neighbor $n$
               \EndFor
               \State \textbf{return} \texttt{True}
            \Else
               \If{$\mathit{type}=+$}
                  \State $\pi_{\mathit{cur}}\leftarrow$ pathlet in $\Pi_u$ 
                  identified by $M.\texttt{f}$ and with start vertex 
                  $M.\texttt{o}$
                  \State $M_P\leftarrow$ new \textbf{Pathlet} message
                  \State $M_P.\texttt{p}\leftarrow\pi_{\mathit{cur}}$
                  \State $M_P.\texttt{t}\leftarrow t$
                  \State Send $M_P$ to neighbor $M.\texttt{src}$
               \Else
                  \State $M_W\leftarrow$ new \textbf{Withdrawlet} message
                  \State $M_W.\texttt{f}\leftarrow M.\texttt{f}$
                  \State $M_W.\texttt{s}\leftarrow M.\texttt{s}$
                  \State $M_W.\texttt{t}\leftarrow t$
                  \State Send $M_W$ to neighbor $M.\texttt{src}$
               \EndIf
               \State \textbf{return} \texttt{False}
            \EndIf
         \Else
            \State \textit{There is no history entry for the pathlet withdrawn 
               by $M$, therefore $u$ cannot know anything about that pathlet. 
               However, the Withdrawlet must still be forwarded}
            \State Add 
            $\left<M.\texttt{f},M.\texttt{o},M.\texttt{s},M.\texttt{t},
            -\right>$ to $H_u$
            \ForAll{$n\in N(u,S(u),M.\texttt{s})\backslash \{M.\texttt{src}\}$}
               \State Send $M$ to neighbor $n$
            \EndFor
            \State \textbf{return} \texttt{False}
         \EndIf
      \EndFunction
   \end{algorithmic}
\end{algorithm*}

\vspace{1mm}
\textbf{Receipt of a Withdraw Message} --
Receiving a \textbf{Withdraw} message $M$ has the same effect of receiving 
several 
\textbf{Withdrawlet} messages, all with the same timestamp $M.\texttt{t}$, one 
for 
each $\fid$ of the pathlets in $\Pi_u$ that have scope stack $M.\texttt{s}$ 
and start vertex 
$M.\texttt{o}$. In order to handle this type of message, function 
\textsc{IsWithdrawletMessageFresher} needs to be 
slightly modified as follows: if $M$ carries fresher information for all the 
pathlets in $\Pi_u$ with scope stack $M.\texttt{s}$ and start vertex 
$M.\texttt{o}$, then history $H_u$ is 
appropriately updated for all these pathlets and only the single 
\textbf{Withdraw} message is further propagated by $u$; otherwise, if $u$ has a 
more recent history entry in $H_u$ for at least one of these pathlets, $u$ 
treats the \textbf{Withdraw} message exactly as a sequence of 
\textbf{Withdrawlet} messages, sending back to $M.\texttt{src}$ single 
\textbf{Pathlet} and \textbf{Withdrawlet} messages with updated information, and 
forwarding single \textbf{Withdrawlet} messages as appropriate. If $M$ is 
determined to carry fresh information, pathlets are then updated by $u$ as 
already explained for the \textbf{Withdrawlet} message.

\section{Applicability Considerations}\label{sec:applicability}
In this section we describe how the control plane we have formally defined can
be implemented in real world, and we explain how further requirements, like 
support for Quality of Service levels, can easily be accommodated in our model.
It is out of the scope of this paper to detail the configuration language
that would have to be used to configure our control plane.

\textbf{Technologies} --
The control plane we have defined in the previous sections is completely
independent of the data plane that carries its messages: network
destinations carried by \textbf{Pathlet} messages are completely generic and
each vertex only communicates with its immediate neighbors, and to achieve this
a simple link-layer connectivity is required. However, the information collected
by our control plane can only be fully exploited if a data plane that can handle
pathlets is available. As also explained in 
Section~\ref{sec:dissemination-basics}, data packets
should have an additional header that specifies the sequence of FIDs of the
pathlets that the packet is to be forwarded along. When a router receives a
packet, it looks at the topmost FID, retrieves the next-hop router that
corresponds to that FID, removes the FID from the packet's header,
and forwards the packet to the next-hop router. If the FID corresponds to a
crossing or final pathlet, the router also alters the packet's header by
prepending the FIDs of the component pathlets before forwarding it. We
highlight that the sequence of FIDs can be represented by a stack of labels and
the operations we have described actually correspond to a label swap. For this
reason, it is easy to implement the data plane of pathlet routing, as well as 
our control plane, on top of MPLS.
The authors of~\cite{bib:pathlet} share the same vision 
in~\cite{bib:pathlet-web}, yet they underline that MPLS does not allow to 
implement an overlay topology, which is 
very useful to specify, e.g., local transit policies. We argue that, 
unlike~\cite{bib:pathlet}, our control plane is conceived for internal routing 
in an ISP's network, a different scenario where MPLS is a commonly adopted 
technology and different requirements exist in terms of routing policies.

\textbf{Incremental Deployment} --
It is of course unrealistic for an Intenet Service Provider to change the
internal routing protocol in the whole network in a single step. Our control
plane is therefore designed to support an incremental deployment, so that a
pathlet-enabled zone of the network that adopts our control plane and an MPLS
data plane can nicely coexist with other non-pathlet-enabled zones of the same
network that use different control and data planes. Assuming that
non-pathlet-enabled zones use IP (possibly in combination with MPLS), we have
two interesting situations: a pathlet-enabled zone is embedded in an IP-only
network (initial deployment phase) or an IP-only zone is embedded in a
pathlet-enabled network (legacy zones that may remain after the deployment). The
first scenario can be easily implemented by making routers at the boundary of
the two zones redistribute IP prefixes from the IP control plane to the pathlet 
control plane: this means that
boundary routers appear as the originators of these destination prefixes in the 
pathlet
zone. Each boundary router then creates final pathlets to get to the
destinations originated by the other boundary routers: the IP prefixes that
boundary routers learn from the pathlet zone in this way are then redistributed
from the pathlet control plane to the IP control plane. Likewise, the IP 
prefixes of destinations that are 
available at routers within the pathlet zone are also redistributed to the 
IP control plane. In this way, IP-only routers can reach destinations inside the 
pathlet zone or just traverse it as if it were a network link.
As a small exception to what we have shown in Section~\ref{sec:model}, in
this scenario boundary routers need to compose final pathlets even if they just
belong to area $A_{(l_0)}$. From the point of view of the data plane, packets 
that enter the pathlet-enabled zone will have suitable $\fid$s pushed on their 
header, indicating the pathlets to be used to reach another boundary router or a 
destination within the pathlet zone; these $\fid$s will be removed when packets 
exit the pathlet-enabled zone.
The second scenario can be implemented by assuming that routers at the boundary
of the two zones have a way to exchange the pathlets they have learned
by exploiting the IP-only control plane: for example, this could be
achieved by tunneling \textbf{Pathlet} messages in IP or by
transferring pathlet information suitably encoded in BGP messages (possibly in
the AS path attribute). Boundary routers then redistribute from the pathlet 
control plane to the IP control plane the IP prefixes they have learned from the 
final pathlets: in this
way, the boundary routers appear as the originators of these prefixes in the
IP-only zone. Moreover, each boundary router will disseminate final
and crossing pathlets that lead, respectively, to destinations inside the IP-only
zone or to other boundary routers. In this way, the IP-only 
zone can be
traversed (or its internal destinations be reached) without revealing its
internal routing mechanism, and appears just as if it were an area of the
pathlet zone. From the point of view of the data plane, a packet containing 
$\fid$s in 
its header must be enabled to traverse the IP-only zone: this can be easily 
achieved by establishing tunnels between pairs of boundary routers.
Of course there is no sharp frontier between the first and the second scenario,
because the roles of ``embedded'' and ``embedder'' zone can be easily swapped:
although they best fit specific phases of the deployment, 
both choices can indeed be permanently adopted, and it is up to the
administrator to decide which one is it best to apply.

\textbf{Quality of Service} --
We have designed our control plane to support the computation of multiple paths
between the same pair of routers. Besides improving robustness, this feature
can also be exploited to support Quality of Service. In particular, each
pathlet could be labeled with performance indicators (delay, packet loss,
jitter, etc.) that characterize the quality of the path that it exploits. Upon
creating a crossing or final pathlet, a router will
update the performance indicators according to those of the component
pathlets. When multiple pathlets are available between the same pair of routers,
a router will be able to choose the one that best fits the QoS requirements for
a specific traffic flow.

\textbf{Software Defined Networking} --
A relatively recent trend in computer networks is represented by the separation
of the logic of operation of the control plane of a device from the (hardware
or software) components that take care of actual traffic
forwarding. This trend, known as Software Defined Networking, has a concrete
realization in the OpenFlow protocol
specification~\cite{web-openflow}. We believe that our approach has several
elements that make it compatible with an OpenFlow scenario. First of all,
the fact that packets are forwarded according to the sequence of FIDs contained
in their header is a form of source routing: this matches with the
OpenFlow mechanism of setting up flow table entries to route all the packets of
a flow along an established path. Moreover, a recent
contribution~\cite{openflow-hierarchical} proposes a
hierarchical architecture for an OpenFlow network: the authors suggest that 
a set of devices under the coordination of a single controller can be seen
as a single logical device that is part of a larger OpenFlow network, in turn
having its own controller. Following this approach, we could assign an OpenFlow
controller instance to each area defined in our control plane, and
these instances could be organized in a hierarchy that simply reflects the
hierarchy of areas: in this way, each instance can direct traffic along the 
desired sequence of pathlets within the area that it controls, whereas instances 
at higher levels of the hierarchy can only see
lower controllers as a single entity, reflecting the idea of crossing pathlet.

\section{Experimental Evaluation}\label{sec:experiments}

In order to verify the effectiveness of our approach and to assess its
scalability, we have performed several experiments in a simulated scenario. For
this purpose we used OMNeT++~\cite{omnet}, a component-based C++ simulation
framework based on a discrete event model. We considered a few other alternative
platforms, including, e.g., the Click modular router~\cite{click}, but in the
end we selected OMNeT++ because it has a very accurate model of a router's
components, like Click, and it also allows to run on a single machine a complete
simulated network with realistic parameters, such as link delay. Moreover, there
exist lots of ready-to-use extensions for OMNeT++ that allow the simulation of 
specific scenarios,
including IP-based networks. To consider a realistic setup, we therefore chose
to build a prototype implementation of our control plane based on the IP 
implementation made available in the INET framework~\cite{bib:inet}, a companion 
project of OMNeT++. In our prototype, the messages of our control
plane are exchanged encapsulated in IP packets with a dedicated protocol number
in the IP header. We implemented most of the mechanisms
described in Sections~\ref{sec:dissemination-basics} and~\ref{sec:dynamics},
with very few exceptions that are not relevant for the purposes of our
experiments. In particular, we implemented all the message types (except fields 
carrying network destinations), all the propagation conditions, the mechanisms 
to discover border vertices for an area and to compose atomic and crossing 
pathlets, the history at each vertex, and a relevant portion of the forwarding 
state (the mapping between a pathlet and its component pathlets). Some of 
the algorithms adopted in our implementation may
still not be tuned for best efficiency, but this is completely irrelevant
because we measured routing convergence times by using the built-in OMNeT++
timer, which reflects the event timings of the simulation,
instead of the wall clock.

Each simulation we ran had two inputs: a topology specification, consisting of 
routers, links, assignment of label stacks to routers, and link delays; 
and an IP routing specification, consisting of assignments of IP addresses to 
routers' interfaces and of insertion of static routes for the networks that were 
directly connected to each router. In order to facilitate the automated 
generation of large topologies, we assigned to each link a /30 subnet selected 
according to a deterministic but completely arbitrary pattern.

We first executed several experiments in a small topology with a well-defined
structure encompassing border routers for several areas.
This topology consisted of 15 routers, 20 edges, 4 areas with a maximum length 
of the label stacks equal to 3 (including label $l_0$), and at least 3 vertices 
in each area.
This helped us to thoroughly verify the implementation for consistency. We then 
implemented a topology generator and used it to create larger topologies that 
could allow us to assess the scalability of 
our control plane. The topology generator works by creating a hierarchy of areas 
and by adding routers and links randomly to the areas. It takes the following 
parameters as input:
length $N$ of the label stack of all the routers;
number of routers having a stack of length $N$, specified as a range
$[R_{\min}, R_{\max}]$, with possibly $R_{\min}=R_{\max}$; 
number of areas contained in each area, specified as a
range $[A_{\min}, 
A_{\max}]$,  with possibly $A_{\min}=A_{\max}$;
probability $P$ of adding an edge between two vertices;
fraction $B$ of the routers within an area that act as border routers 
for that area (namely that can have links to vertices outside that area). 
The topology generator proceeds by recursively creating areas, 
starting from a single area that comprises all vertices. The complete procedure 
is described in Algorithm~\ref{alg:topology-generator}.

A detailed description of the experiments we carried on follows. This 
description is still in a drafty form and will be improved in a future release 
of this technical report.

Two preliminary experiments were carried out to assess the scalability of our 
control plane.
In both experiments, we ran several simulations where the size of the topology 
were increased. 
 To increase the size of the topology we adopted the following strategy. 
Every input parameter of the topology generator was fixed, except one that has 
been used to change the size of the topology. In the first experiment, the 
number of areas contained in each area was chosen as the variable input, while  
in the second experiment, the length of the label stack was chosen as the 
variable input.
For each simulation, we collected data regarding the number of messages sent by 
each router, the number of pathlets stored in each router, and the convergence
time of the protocol. Because of difficulty with the OMNeT framework, our simulation
were performed with topologies with a limited level of multipath. As a consquence,
we ran simulations on topologies whose bottom-level areas exposes a limited
level of multipath. On average, there are $3-4$ different paths between two 
border vertices of a bottom-level area. 

{\bf Statistical tools -} The analysis of the collected data involves the 
knowledge of several widely adopted statistical tool. We introduce them and we 
try to give intuitions of their meaning.
 We use linear regression analysis to determine the level of correlation 
(linear dependence) between two arbitrary variables $X$ and $Y$. A linear 
regression analysis returns a line $l_Y(X)=A+BX$ that is an estimate of the 
value $Y$ with respect to $X$. The estimation of $l_Y$ depends on the specific
measure of accuracy that is adopted. In our analysis, we use the 
Ordinary-Least-Square (OLS) method for estimating coefficients $A$ and $B$ of 
$l_Y$. Observe that $B$ is extremely relevant in the analysis of the result 
since it can be interpreted as an estimation of the increase of $Y$ for each 
additioanl unit of $X$. 
 To verify if there exists a linear relation between $X$ and $Y$, we 
look at the \emph{coefficient of determination} $R^2$. An $R^2$ close to $1.0$ indicates 
that the relation between $X$ and $Y$ is roughly linear, while an $R^2$ closer 
indicates that the relation does not seem to be linear. To determine the level
of dispersion of the observed values for $X$ and $Y$ with respect to $l_Y$, we 
look at the \emph{standard error of the regression} ($SER$) of $l_Y$. This value has the
same unit of values in $Y$ and can be interpreted as follows: approximately 95\% 
of the points lie within $2\times SER$ of the regression line. Sometimes, it is
better to look at a normalized value that represents the level of dispersion of 
a set of values $X$. This can be done, by computing the \emph{coefficient of 
variation} $c_v$ of $X$ which is independent of the unit in which the measurement has 
been  taken and can be interpreted as follows: given a set of values $X$, 
approximately 68\% of the values lie between $\mu(1 - c_v)$ and $\mu(1 + c_v)$,
approximately 95\% of the values lie between $\mu(1 - 2c_v)$ and $\mu(1 + 2c_v)$,
and approximately 99.7\% of the values lie between $\mu(1 - 3c_v)$ and $\mu(1 + 3c_v)$,
where $\mu$ is the arithmetic mean of $X$.

{\bf Experiment 1 - } We constructed network 
topologies using the following fixed parameters: $R_{\min}=R_{\max}=
10$, $N=2$, $P=0.1$ and $B=5$. 
The variable input $A_{\min}=A_{\max}$ varied from $2$ to $7$. Roughly speaking,
we keep a constant number of levels in the area hierarchy while increasing the
number of areas in the same level. For each 
combination of these values, we generated $10$ different topologies and for each 
of these topologies, we ran an OMNeT++ simulation and collected relavant data as
previously specified. 
 In Fig.~\ref{fig:area-max-pathlet} (Fig.\ref{fig:area-avg-pathlet}) we show that the 
maximum (average) number of pathlets stored in each router (depicted as crosses)
grows with respect to the number of edges in the topology. In particular, the growth 
is approximately linear as confirmed by a linear regression analysis (depicted as a line) 
performed on the collected data. The slope of the line is $0.83$ ($0.75$), which can be 
interpreted as follows: for each new edge in the network, we expect that the value 
of the maximum (average) number of pathlets stored in each router 
increases by a factor of $0.83$ ($0.75$) on average. To assess the linear
dependence of the relation, we verified that the $R^2=0.87$ ($R^2=0.9$), which means that 
the linear regression is a good-fit of the points. The standard error of the 
regression is $39.8$ ($9.26$), which can be interpreted as follows: 
approximately 95\% of the points lie within $2\times 39.8$ ($2\times 9.26$) of 
the regression line. 
 We motivate the linear growth as follows. Observe that, each topology has a fixed 
number $A_{\max}$ of areas and therefore the expected number of crossing pathlets 
created for any arbitrary area is the same. Hence, by increasing 
$A_{\max}$, since the number of crossing  pathlet created in each area grows 
linearly with $A_{\max}$, we have that also the number of pathlets stored in 
each router, which contains a constant number of atomic pathlets plus each crossing 
pathlet created by border router of neighbors areas, grows linearly. 
In Fig.~\ref{fig:area-max-msg} (Fig.\ref{fig:area-avg-msg}) we show that similar
results hold also when we analyze the average/maximum number of messages by each 
router. In fact, we observed that many considerations that we observed with 
respect to the number of messages sent by each router, are also valid with 
respect to the number of pathlets stored in each router. In fact, as shown in 
Fig.~\ref{fig:area-max-msg-max-pathlet}, we show that there exists an interesting 
linear dependence, with $R^2=0.87$, between the number of messages sent by each router and the 
number of pathlets stored in each router. 

{\bf Experiment 2 - } We constructed 
network topologies using the following fixed values:
$R_{\min}=R_{\max}=10$, $A_{\min}=A_{\max}=2$, $P=0.1$ and $B=5$. The variable 
input $N$ varied  in the range between $1$ and $4$. Roughly speaking,
we keep a constant number of subareas inside an area, while we increase the 
level of area hierarchy. For each combination of these  values, we generated 
$10$ different topologies and for each of these topologies, we 
run an OMNeT++ simulation and the same data as in the first experiment.
 In Fig.~\ref{fig:stack-max-pathlet} (Fig.\ref{fig:stack-avg-pathlet}) we see that the 
maximum (average) number of pathlets stored in each router (depicted as crosses)
grows linearly with respect to the number of edges in the topology. A linear 
regression analysis computed over the collected data shows that the slope of the 
line is $1.38$ ($0.94$), which can be 
interpreted as follows: for each new edge in the network, we expect that the value 
of the maximum (average) number of pathlets stored in each router 
increases by a factor of $1.38$ ($0.94$) on average. To assess the linear
dependence of the relation, we verified that the $R^2=0.75$ ($R^2=0.81$), which 
can still be considered a good-fit of the points. 
As for the dispersion of the points with respect to the line, it is easy to 
observe that, the higher the number of edge, the higher the
dispersion. If we look to the standard error of the regression, since the measure
is not normalized, we may obtain an inaccurate value of the dispersion. Therefore, 
we do the following. We consider $4$ partition $P_1,\dots,P_4$ of the measurents,
where $X_i$ contains each measure collected when $N=i$, and compute the 
coefficient of variation $c_v$ of each subset. We observe that $c_v$ varies
between $0.21$ and $0.42$ ($0.24$ and $0.31$), which means that in each partition $X_i$, 95\% of 
the points lie within $2\cdot0.41\mu$ ($2\cdot0.31\mu$) of the mean $\mu$ of 
$X_i$. We have not enough data to check whether there exists a lineaer dependence 
between the value $c_v$ of a partition $X_i$ and the index $i$. 

We motivate the linear growth as follows.
Observe that each bottom area construct on average the same number of crossing 
pathlets, regardless the levels of the hierarchy. Because of the low values 
chosen for $P$, $A_{\min}$, and $A_{\max}$, the number of crossing pathlets
created by higher areas is roughly proportional to the number of crossing pathlets
created by its subareas. Now, consider the set of pathlets stored in a router with
stack label $(x_1\ \dots\ x_n)$. 
It contains atomic pathlets for the areas in which it belongs, which are logarithmic
with the number of edges, and it contains crossing pathlets from other areas, which
are roughly proportional to the number of bottom-level areas. 
Hence, each time $N$ is increased, 
both the number of edges and the number of bottom-level areas, grows exponentially 
with the same rate, and therefore the number of pathlets stored in each router
increases linearly with respect to the number of edges.
As for the first experiment, a similar trend can be seen also in 
Fig.~\ref{fig:stack-max-msg} and  Fig.\ref{fig:stack-avg-msg} with respected to
the maximum and average number of messages sent by a router, respectively. In 
fact, also in this experiment, we observe that there is a good linear relation
between the number of messages sent by a vertex and the number of pathlets stored
in each router (see Fig.~\ref{fig:stack-max-msg-max-pathlet}).

{\bf Convergence time - } We now consider the convergence time $T$ of the
protocol expressed in milliseconds. In both experiments we set link delays as 
random uniform variable in the range between $10$ and $50$ milliseconds. 
Convergence time for the first and the second experiments with respect to the
size of the network (expressed by the number of edges) are shown in 
Fig.~\ref{fig:area-convergence-time} and \ref{fig:stack-convergence-time}, 
respectively.
In {\bf Experiment 1} we achieved $T\in [363,611]$
whereas in {\bf Experiment 2} $T\in [259,736]$. The minimum value in both
the intervals correspond to the minimum value among the 10 different generated
topologies with the value $A=2$ and the value $L=1$ respectively for {\bf
Experiment 1} and {\bf Experiment 2}; on the other hand the maximum value in
both the intervals is the maximum value among the 10 different generated
topologies with the value $L=7$ and the value $L=4$ respectively for {\bf
Experiment 1} and {\bf Experiment 2}.

In the first experiment, the value $T=363\ \mathit{msec}$ was achieved
for a topology with two areas, $A_{(0\ 1)}$ and $A_{(0\ 2)}$, both contained 
into area $A_0$, with $R_{\min}=R_{\max}=10$ vertices per area. By random adding
edges, we obtained a network with 20 vertices and 22 edges. On the other hand, 
the value $T=611\ \mathit{msec}$ was achieved for a topology with $R_{\min}=
R_{\max}=10$ vertices per bottom-level area and $A_{\min}=A_{\max}=7$ bottom-level areas contained
into area $A_0$. By random adding edges, we obtained a topology with 70
vertices and 89 edges.
As for the experiments that involve the lenght $N$ of the label stack, we obtained
the value $T=0.259\ \mathit{msec}$ for a network  topology where all vertices belong
only to the same area $A_0$. The generated network has 10 vertices and 12 edges. On the
other hand, the value $T=736\ \mathit{msec}$ was achieved for a network topology with
$R_{\min}=R_{\max}=10$ vertices per bottom-level area and the length of the 
label stack is $N=4$. The generated topology has 80 vertices and 104 edges, 
with vertices organized into a network that have 2 areas at each level. 
In particular, $A_0$ contains two areas, $A_{(0\ 1)}$ and $A_{(0\ 2)}$. Each 
of these areas contains, in turn, two areas: $A_{(0\ 1\ 2)}$ and $A_{(0\ 1\ 3)}$ 
are contained into $A_{(0\ 1)}$, while $A_{(0\ 2\ 3)}$ and $A_{(0\ 2\ 4)}$ 
are contained into $A_{(0\ 2)}$ and so on. 

In both experiments, we achieved a convergence time below $1 \mathit{sec}$ and we observed 
that the correlation between the size of the network and the convergence time, 
does not exhibit a linear behaviour.
In Fig.~\ref{fig:area-convergence-time} it can clearly be observed that the 
regression line does not fit well the points. 
The slope of the regression line is $1.2$, which can be interpreted as
follow: for each new edge in the network, we expect that convergence time increases
of $1.2$ milliseconds. However, we computed an $R^2$ value of $0.26$, which is 
close to $0$ and suggest a lack of correlation between convergence time and 
number of edges. In other word, the convergence time is independent from the 
number of edges. On the other hand, In Fig.~\ref{fig:stack-convergence-time}
a more strong correlation between convergence time and number of edges seems to 
hold. In this case, the slope is $2.9$ with an $R^2$ value of $0.57$. 
This means that the growth appears to be linear but there is a high dispersion of
the points from the line. 
We recall that both experiments has been run with a pathlet composition rule 
that force each border vertex to compose all the possible pathlets to every 
other border vertex of the same area. Changing this rule, we expect that 
convergence times will decrease.

\begin{figure}
   \centering
   \includegraphics[width=\columnwidth]{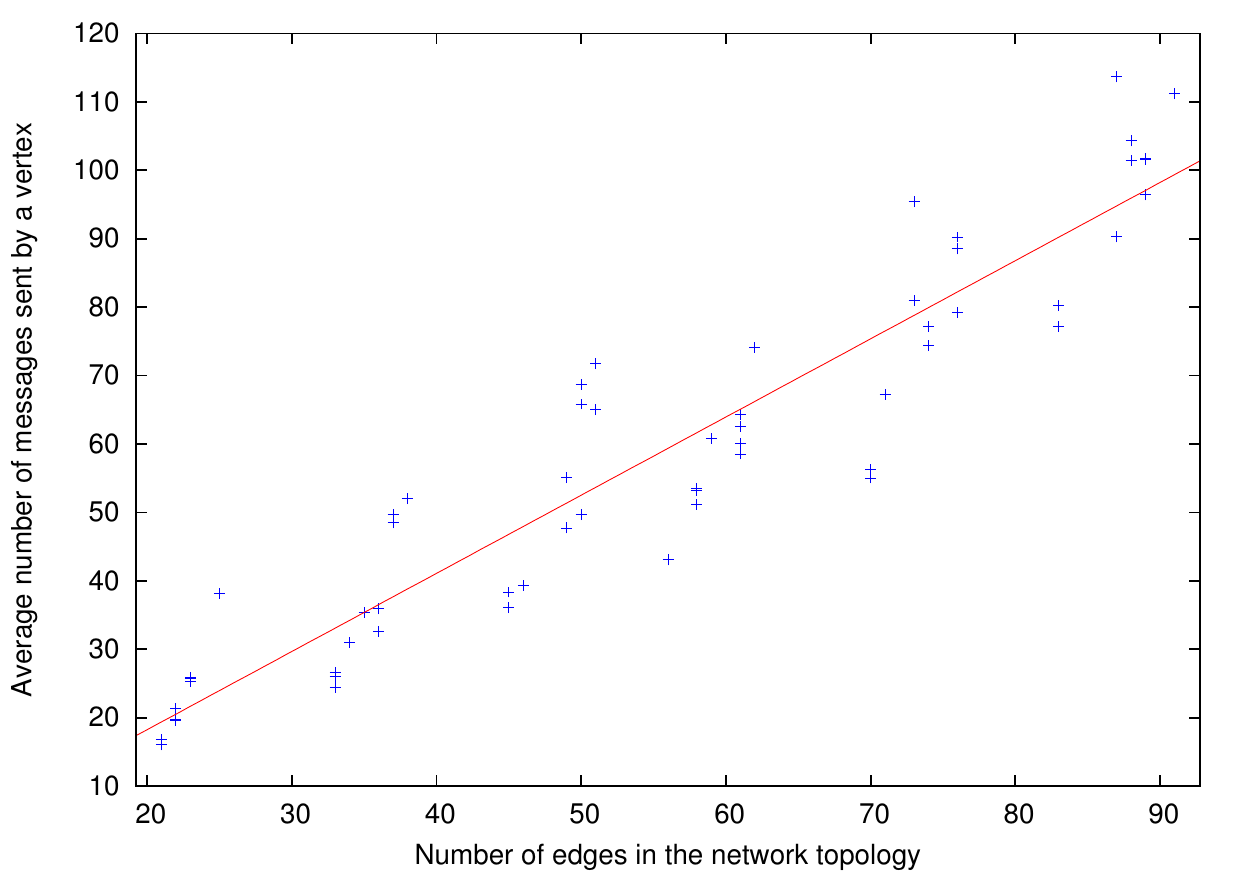}
   \caption{Average number of messages sent by a router with respect to 
   the number of edges contained in the topology. (Experiment 1)}
   \label{fig:area-avg-msg}
\end{figure}

\begin{figure}
   \centering
   \includegraphics[width=\columnwidth]{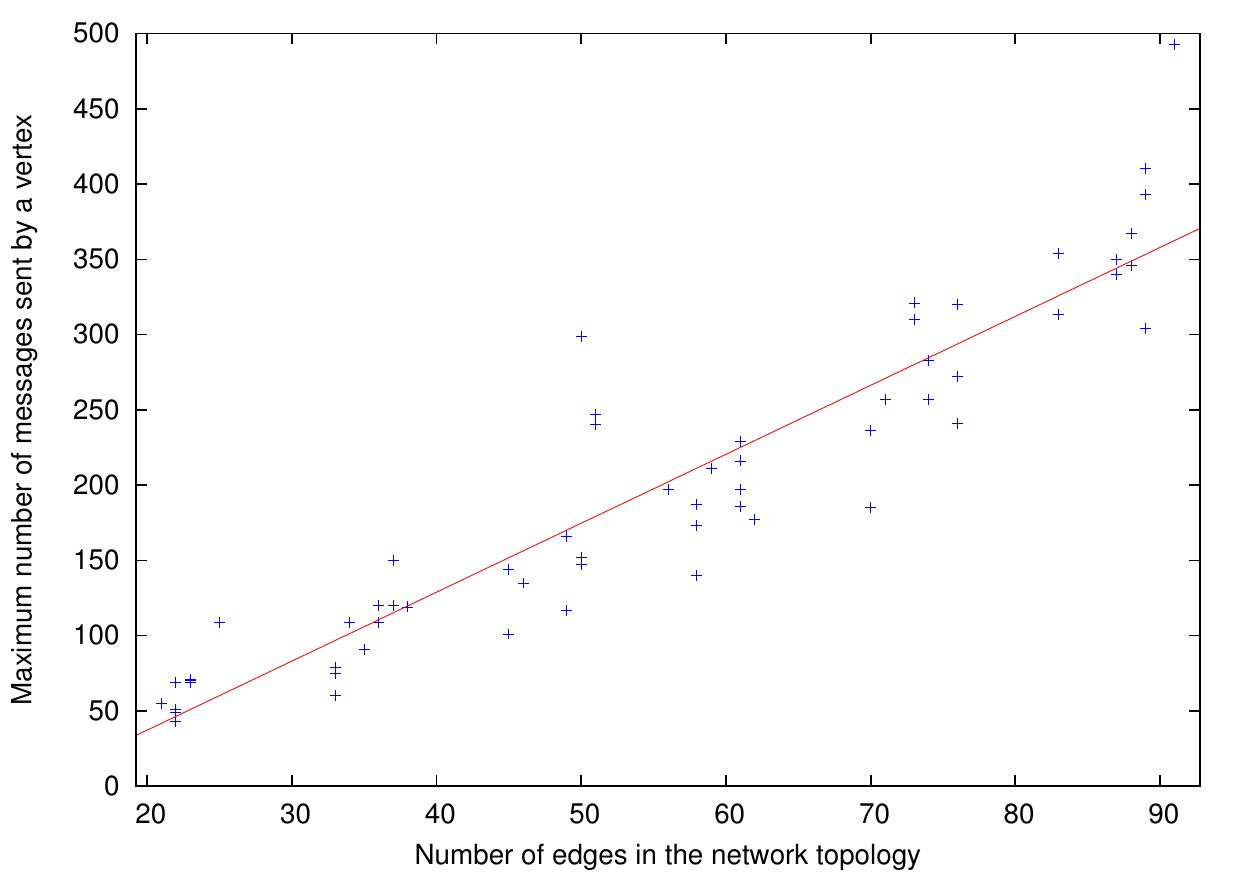}
   \caption{Maximum number of messages sent by a router with respect to 
   the number of edges contained in the topology. (Experiment 1)}
   \label{fig:area-max-msg}
\end{figure}

\begin{figure}
   \centering
   \includegraphics[width=\columnwidth]{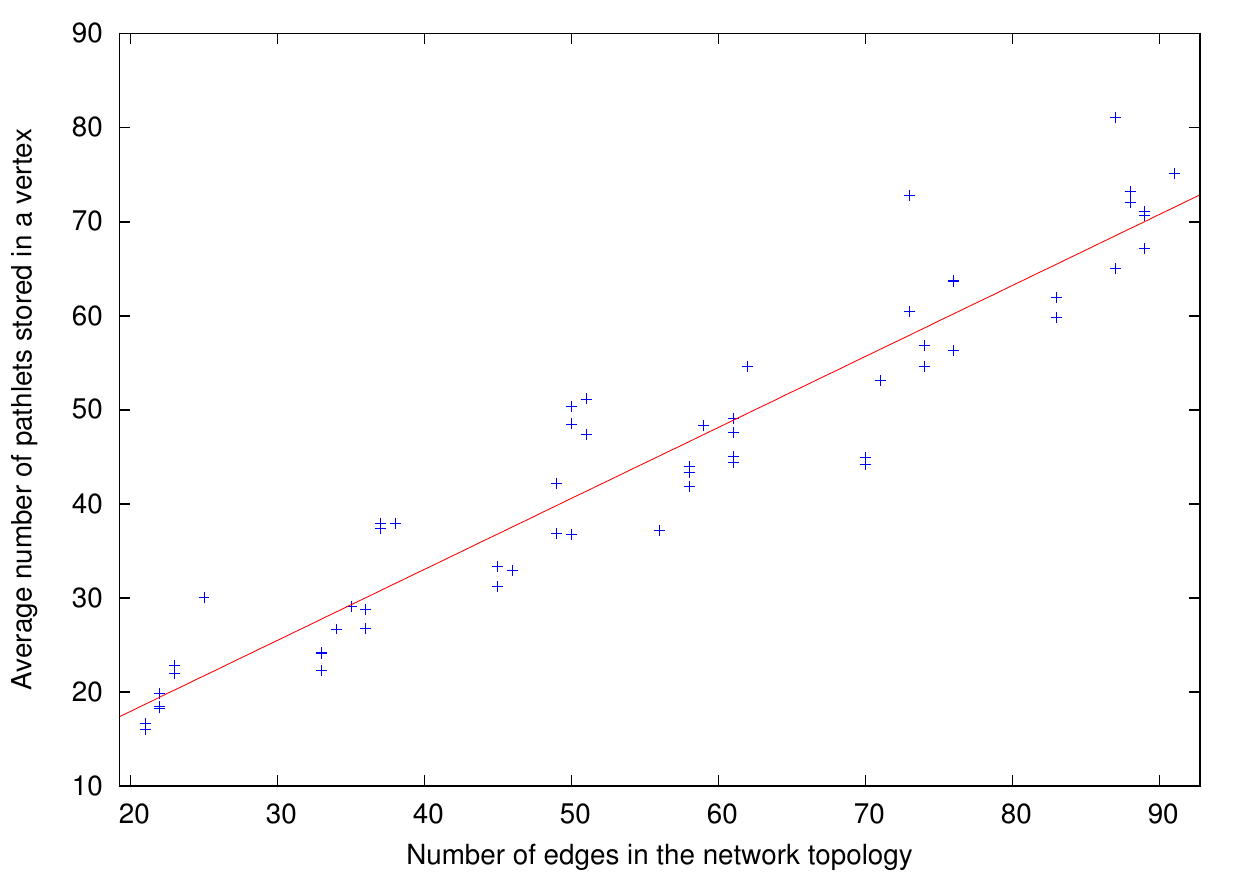}
   \caption{Average number of pathlets stored in a router with respect to 
   the number of edges contained in the topology. (Experiment 1)}
   \label{fig:area-avg-pathlet}
\end{figure}

\begin{figure}
   \centering
   \includegraphics[width=\columnwidth]{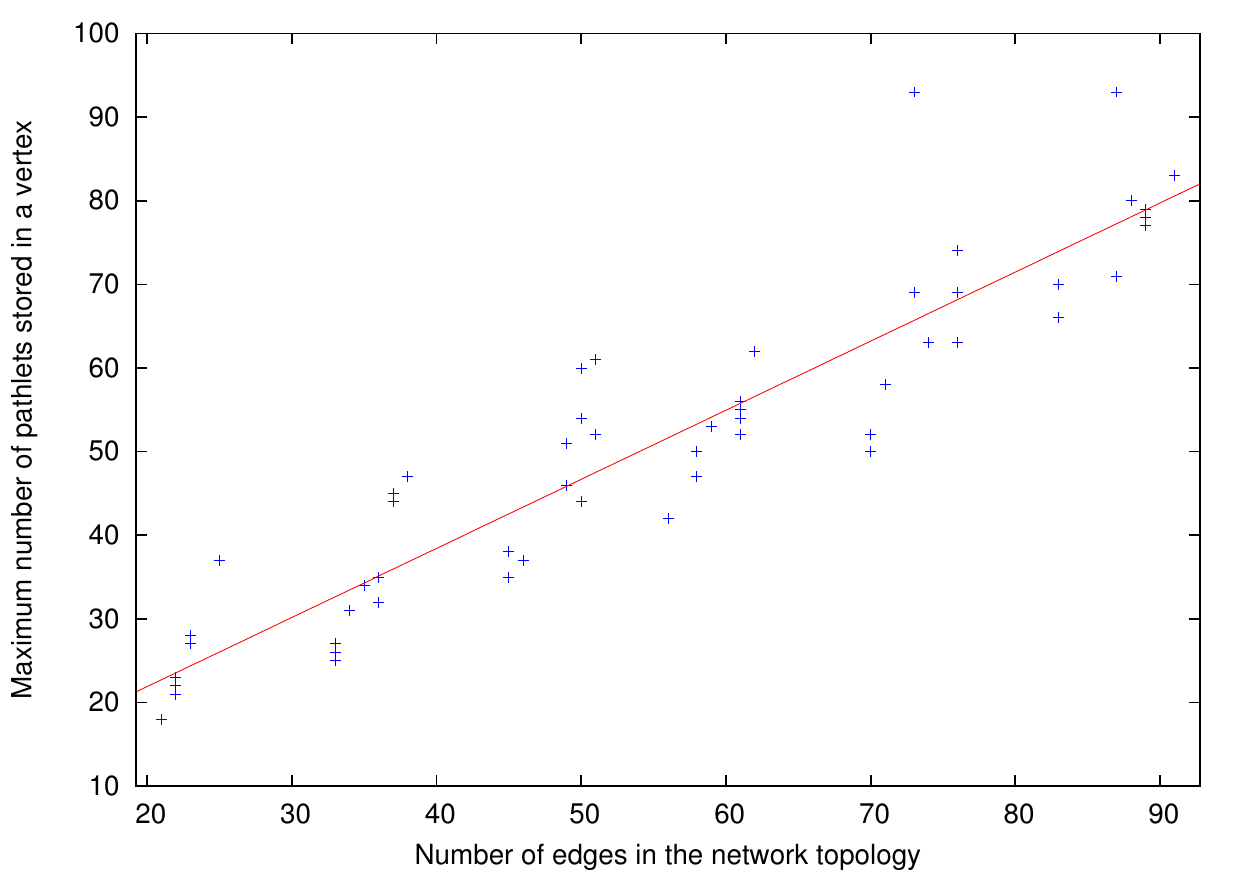}
   \caption{Maximum number of pathlets stored in a router with respect to 
   the number of edges contained in the topology. (Experiment 1)}
   \label{fig:area-max-pathlet}
\end{figure}

\begin{figure}
   \centering
   \includegraphics[width=\columnwidth]{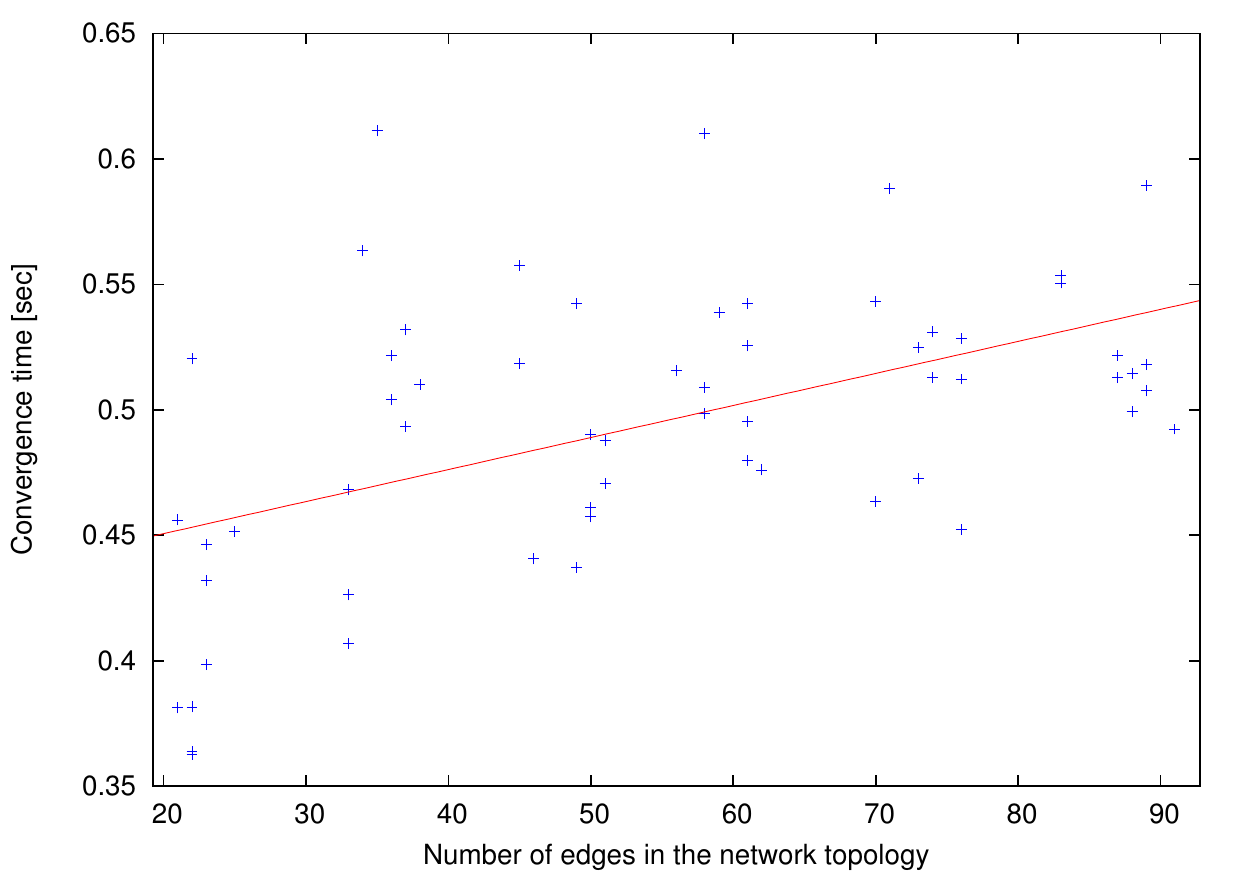}
   \caption{Network convergence time with respect to the number of edges 
   contained in the topology. (Experiment 1)}
   \label{fig:area-convergence-time}
\end{figure}

\begin{figure}
   \centering
   \includegraphics[width=\columnwidth]{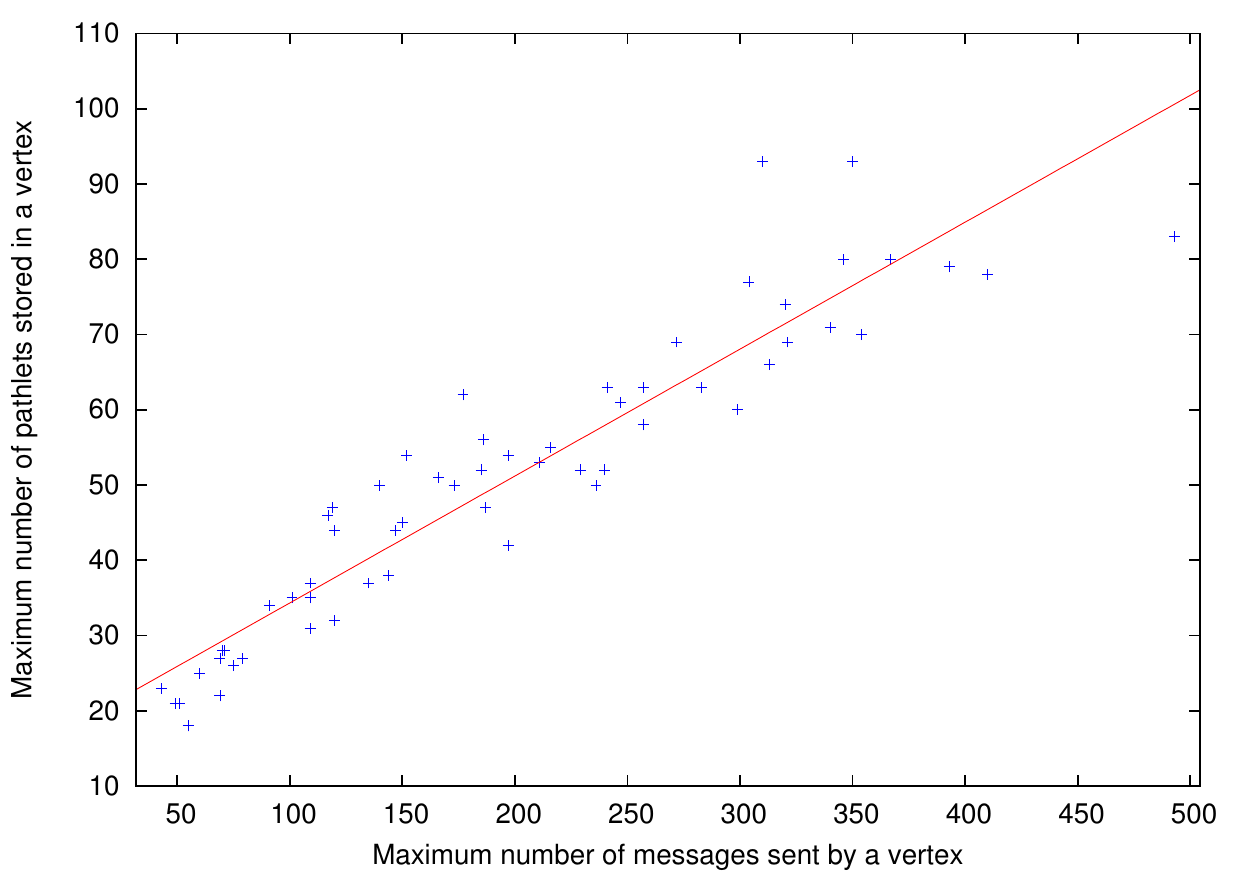}
   \caption{Maximum number of pathlets stored in a router with respect to 
   the maximum number of messages sent by a single router. (Experiment 1)}
   \label{fig:area-max-msg-max-pathlet}
\end{figure}

\begin{figure}
   \centering
   \includegraphics[width=\columnwidth]{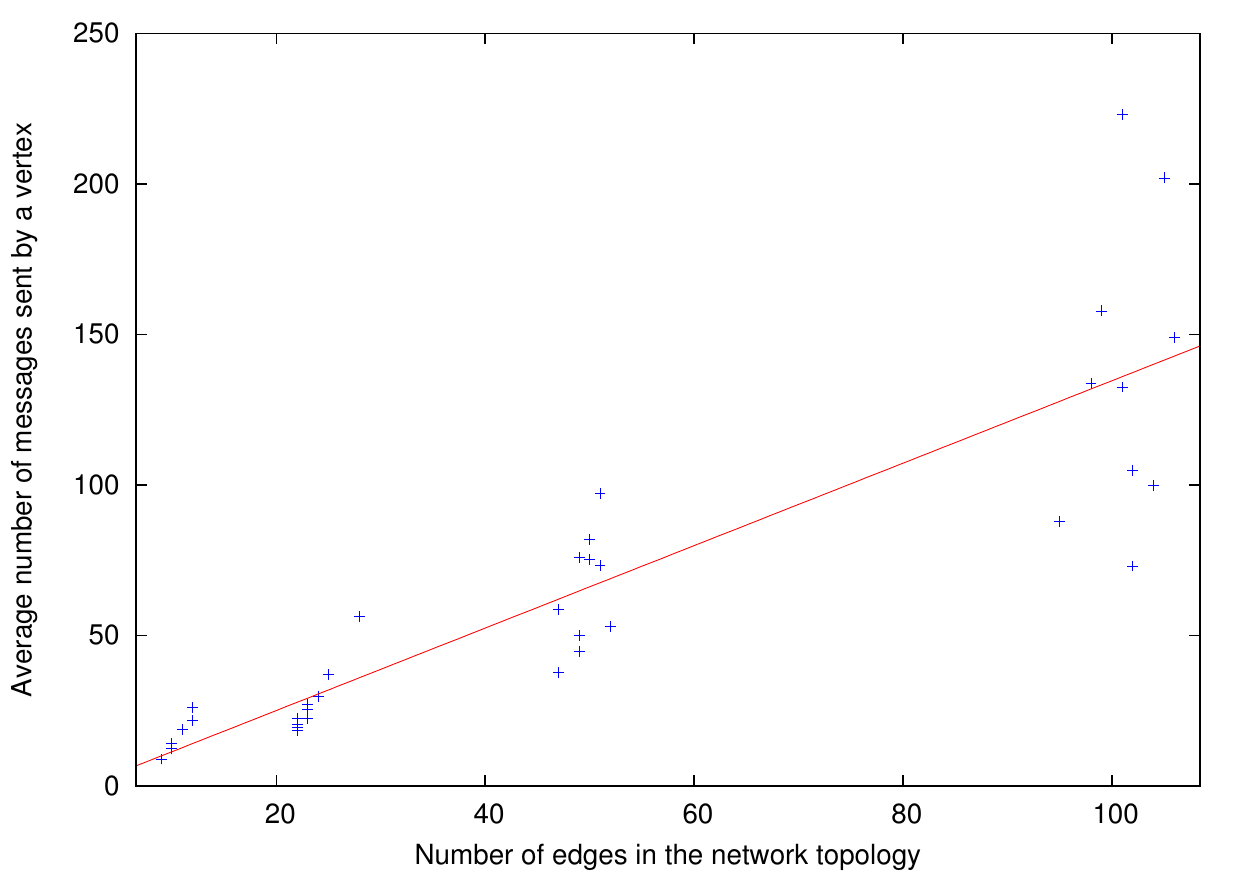}
   \caption{Average number of messages sent by a router with respect to 
   the number of edges contained in the topology. (Experiment 2)}
   \label{fig:stack-avg-msg}
\end{figure}

\begin{figure}
   \centering
   \includegraphics[width=\columnwidth]{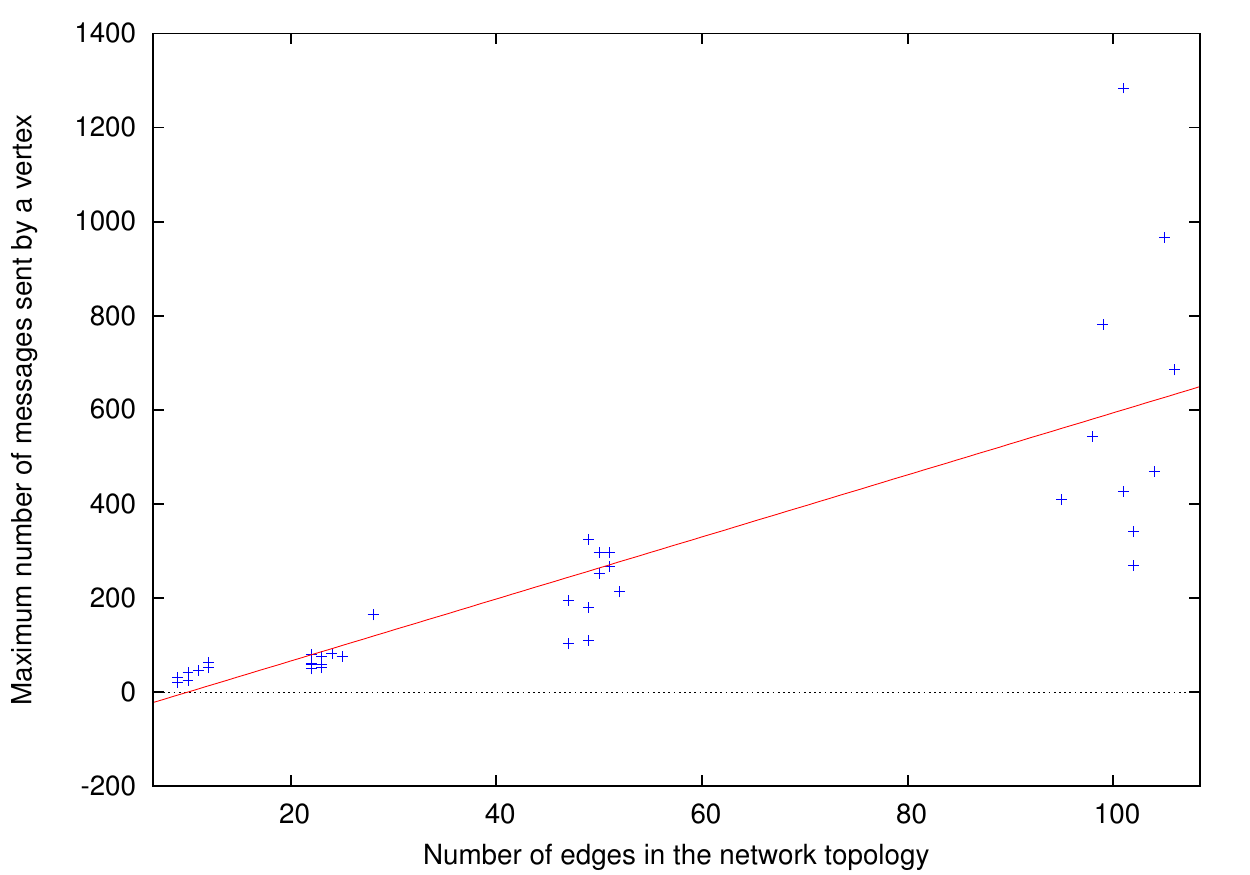}
   \caption{Maximum number of messages sent by a router with respect to 
   the number of edges contained in the topology. (Experiment 2)}
   \label{fig:stack-max-msg}
\end{figure}

\begin{figure}
   \centering
   \includegraphics[width=\columnwidth]{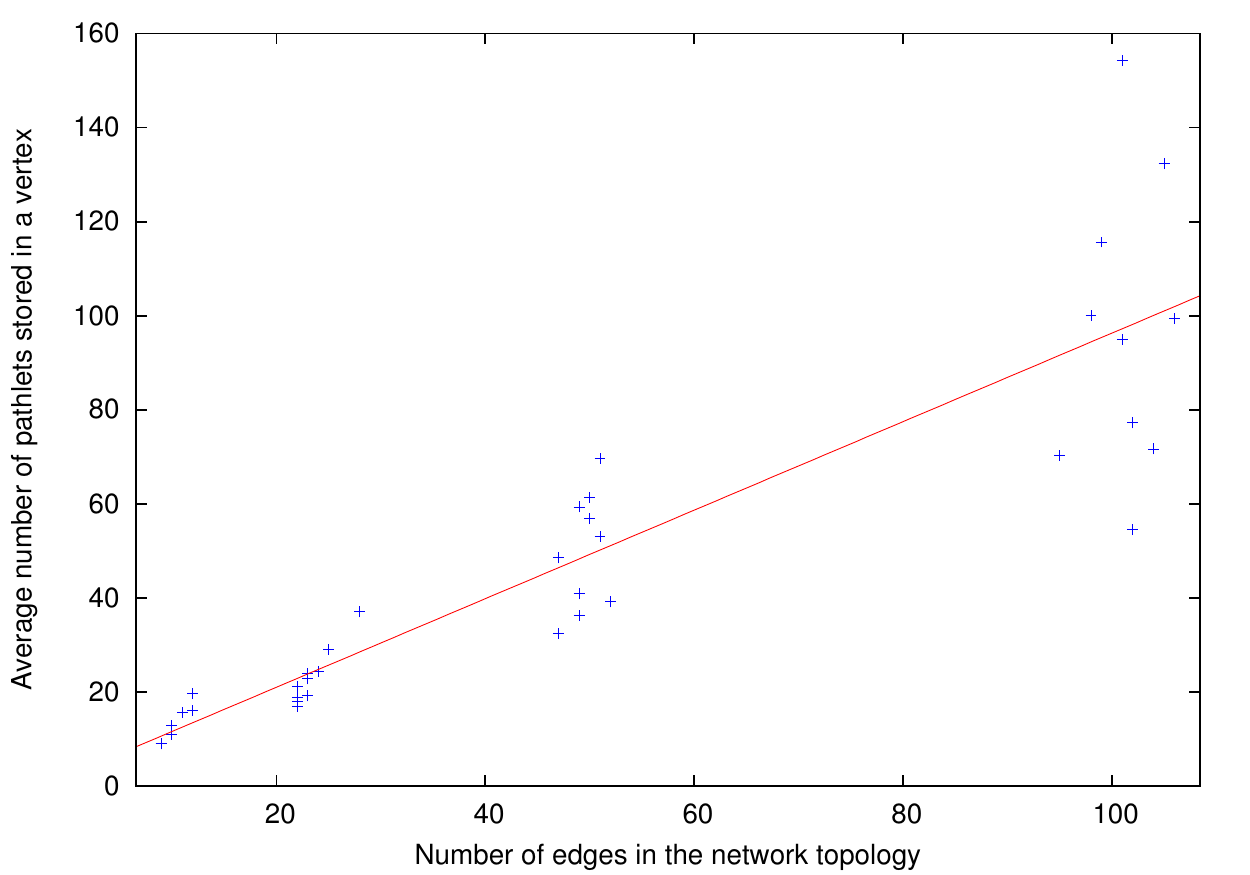}
   \caption{Average number of pathlets stored in each router with respect to 
   the number of edges contained in the topology. (Experiment 2)}
   \label{fig:stack-avg-pathlet}
\end{figure}

\begin{figure}
   \centering
   \includegraphics[width=\columnwidth]{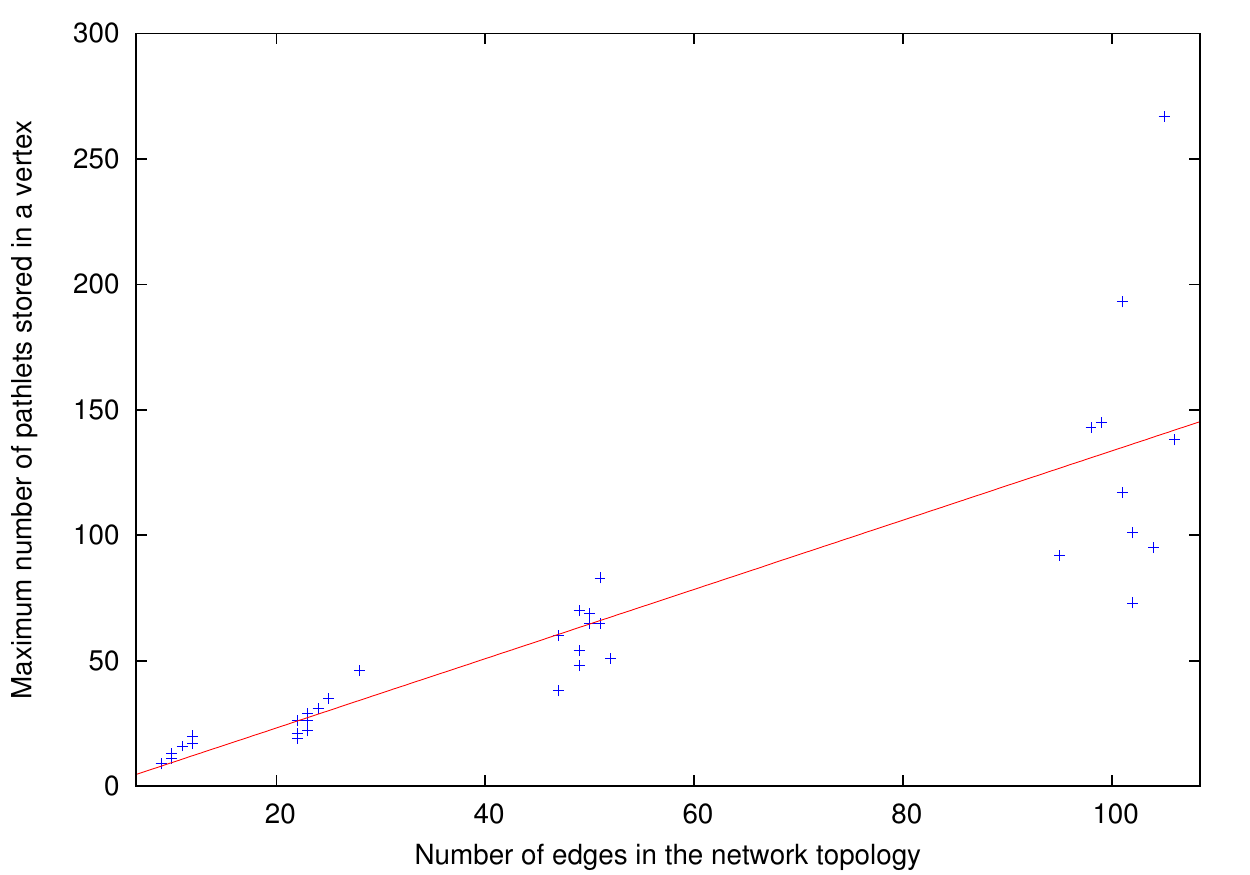}
   \caption{Maximum number of pathlets stored in a router with respect to 
   the number of edges contained in the topology. (Experiment 2)}
   \label{fig:stack-max-pathlet}
\end{figure}

\begin{figure}
   \centering
   \includegraphics[width=\columnwidth]{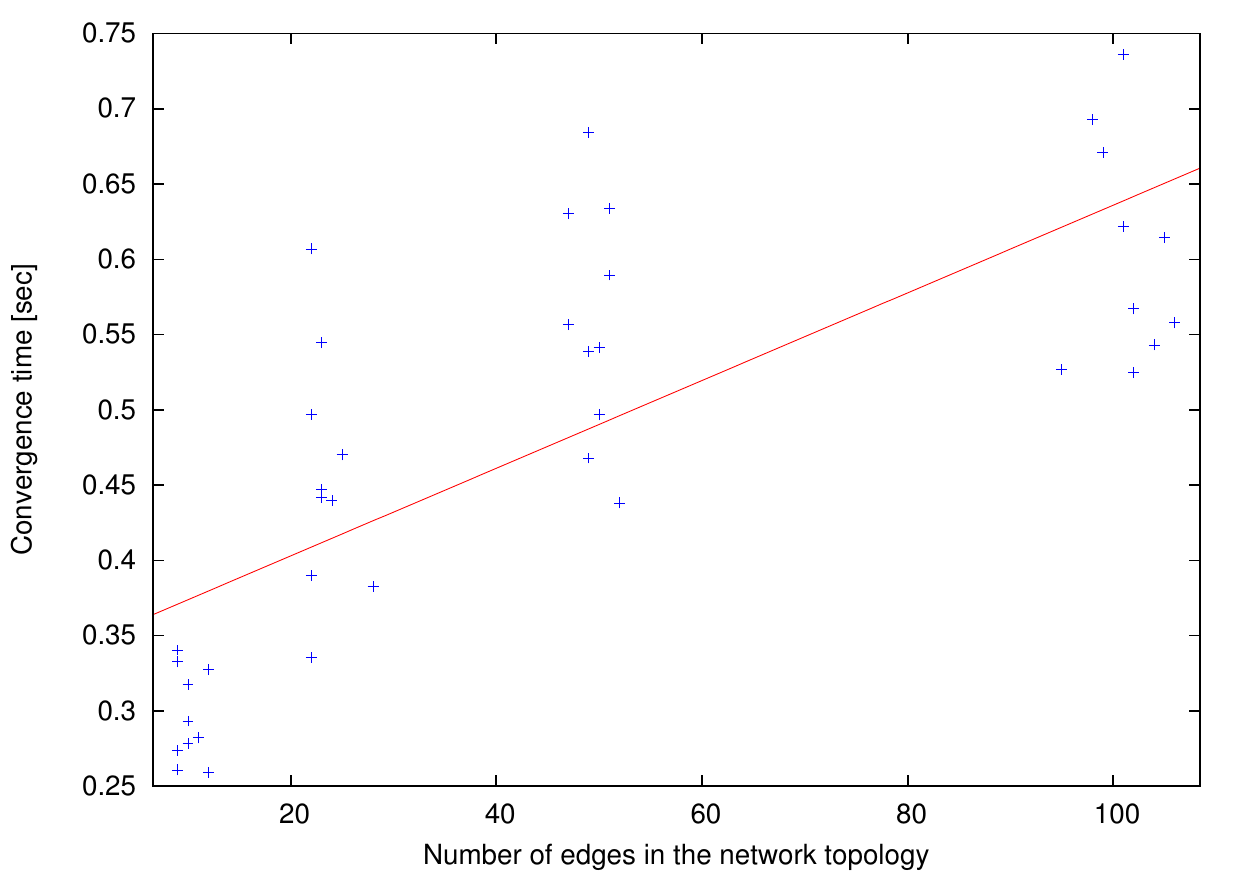}
   \caption{Network convergence time with respect to the number of edges 
   contained in the topology. (Experiment 2)}
   \label{fig:stack-convergence-time}
\end{figure}

\begin{figure}
   \centering
   \includegraphics[width=\columnwidth]{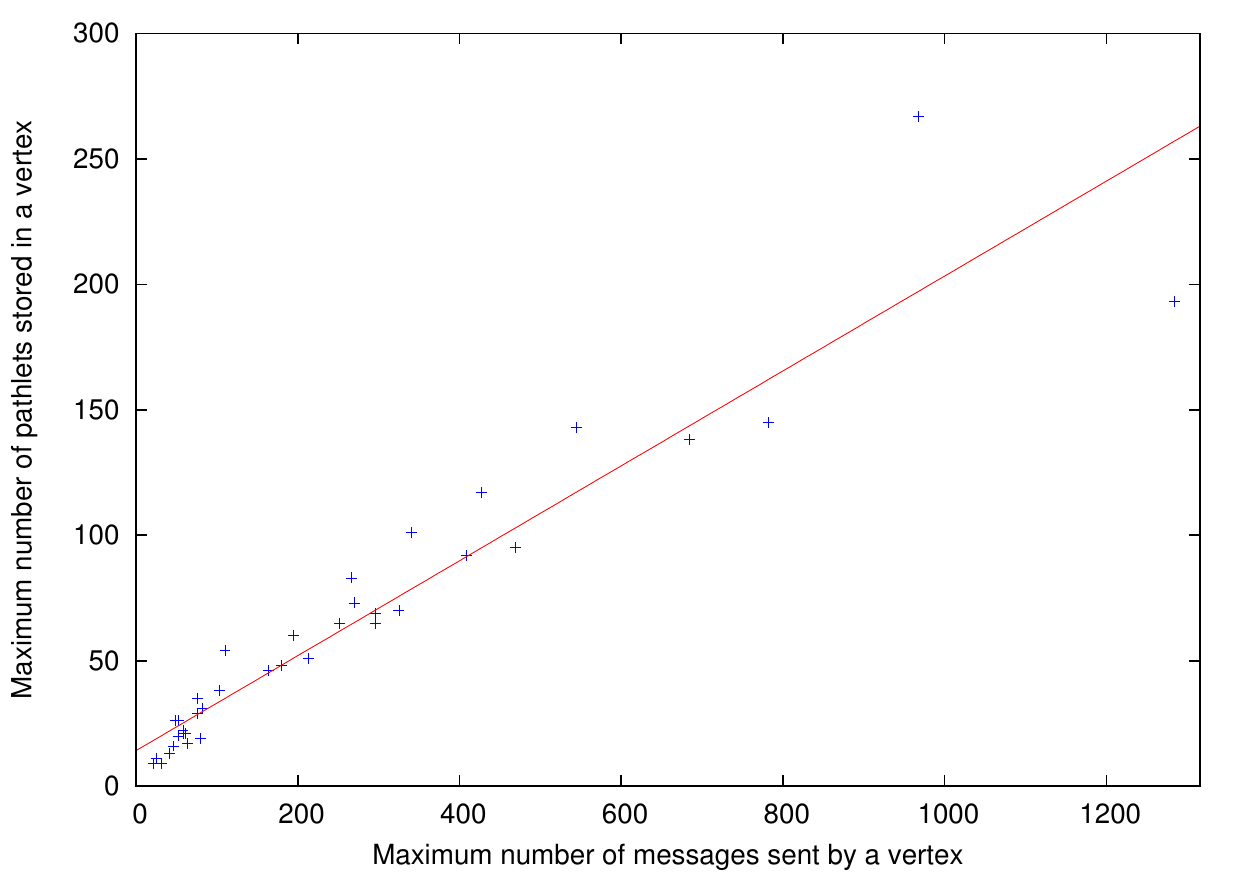}
   \caption{Maximum number of pathlets stored in a router with respect to 
   the number of messages sent by a vertex. (Experiment 2)}
   \label{fig:stack-max-msg-max-pathlet}
\end{figure}


%



\begin{algorithm*}
   \caption{Algorithm used in our topology generator.}
   \label{alg:topology-generator}
   \begin{algorithmic}
      \Function{PopulateArea}{$A$, $\mathit{level}$, 
         $R_{\min}$, $R_{\max}$, $N$, $A_{\min}$, $A_{\max}$, $P$, $B$}
         \If{$\mathit{level}=N$}
            \State $r \leftarrow $ a random number in $[R_{\min}, R_{\max}]$
            \State Add $r$ vertices to $A$
            \Repeat
               \ForAll{pair $(u,v)$ of vertices in $A$}
                  \State Add an edge between $u$ and $v$ with probability $P$
               \EndFor
            \Until{$A$ is connected}
            \State Randomly pick $r\times B$ routers in $A$ and mark them as 
            border routers for $A$
         \Else
            \State $a \leftarrow $ a random number in $[A_{\min}, A_{\max}]$
            \State Create $a$ areas inside $A$; let $\cal A$ be the set of these 
            areas
            \State $\bar R \leftarrow \emptyset$
            \ForAll{$\bar A$ in $\cal A$}
               \State \Call{PopulateArea}{$\bar A$, $\mathit{level}+1$, 
                  $R_{\min}$, $R_{\max}$, $N$, $A_{\min}$, $A_{\max}$, $P$, $B$}
               \State $\bar R \leftarrow \bar R\ \cup$ all border routers for 
               $\bar A$
            \EndFor
            \Repeat
               \State $E \leftarrow \emptyset$
               \ForAll{pair $(u,v)$ of routers in $\bar R$}
                  \State Flip a coin with probability $P$
                  \If{heads}
                     \State Add an edge between $u$ and $v$
                     \State $E \leftarrow E \cup {(u,v)}$
                  \EndIf
               \EndFor
            \Until{the undirected graph formed by vertices in $\bar R$ and edges 
            in $E$ is connected}
            \State Randomly pick $|\bar R |\times B$ routers in $\bar R$ and 
            mark them 
            as border routers for $A$
         \EndIf
         \State \textbf{return} $A$
      \EndFunction
      
      \Function{TopologyGenerator}{$R_{\min}$, $R_{\max}$, $N$, $A_{\min}$, 
      $A_{\max}$, 
         $P$, $B$}
         \State Create an area $A$
         \State $\mathit{level} \leftarrow 1$
         \State \textbf{return} \Call{PopulateArea}{$A$, $\mathit{level}$, 
         $R_{\min}$, $R_{\max}$, $N$, $A_{\min}$, $A_{\max}$, $P$, $B$}
      \EndFunction
   \end{algorithmic}
\end{algorithm*}
      
Our prototype implementation, including the topology generator, is publicly 
available at~\cite{pathlet-omnet-implementation}.


%

\section{Conclusions and Future Work}\label{sec:conclusions}
In this paper we introduce a control plane for internal routing inside an ISP's
network that has several desirable properties, ranging from fine-grained
control of routing paths to scalability, robustness, and QoS support. Besides
introducing the basic routing mechanisms, which are based on a well-known
contribution~\cite{bib:pathlet}, we provide a thorough and formally sound
description of the messages and algorithms that are required to design such a
control plane. We validate our approach through extensive experimentation in
the OMNeT++ simulator, which reveals very promising scalability and convergence
times. Our prototype implementation is available
at~\cite{pathlet-omnet-implementation}.

There are a lot of improvements that we are still interested in working on.
Some of them are possible optimizations, whereas others are foundational issues
that are still open: here we mention a few. Our current choice of messages
types imposes a strong coupling between routing paths and network destinations:
if a network destination changes its visibility (for example, a router starts
announcing a new IP prefix), several pathlets to that destination must be
(re)announced, even though the routing has not changed. Inspired by recent
research trends~\cite{bib:ietf-lisp-group}, we could change the protocol a bit
to separate these two pieces of information.
In line with this decoupling requirement, we would like to investigate on how to
deal with dynamic changes in QoS levels associated with pathlets.
Routing policies, especially pathlet composition rules, could of course be
refined to accommodate further requirements that we have not considered yet.
Moreover, their specification and application could be enhanced to improve
scalability in common usage scenarios (for example, when several
areas are grouped into a larger one). The pathlet expiration mechanism needs 
further improvements to correctly purge pathlets in the presence of routing 
policies.
The handling of stack change events could also be improved: in particular,
we could design more effective mechanisms to transparently replace a pathlet
that is no longer visible with other newly appeared pathlets, without spreading
messages to the whole network. Being modeled as stack change events, the
handling of faults could be improved likewise.

\bibliographystyle{IEEEtran.bst}

\end{document}